\documentclass[10pt, a4paper]{article}

\usepackage[top=2.4cm, bottom=2.4cm, left=2cm, right=2cm]{geometry}
\usepackage{enumitem}
\setitemize{noitemsep,topsep=1pt,parsep=1pt,partopsep=1pt}
\usepackage[utf8]{inputenc}
\usepackage{amsmath,amsfonts,amsthm}
\usepackage{amssymb}
\usepackage{xspace}
\usepackage{graphicx}
\usepackage[caption=false]{subfig}
\usepackage{textcomp}
\usepackage{xcolor}
\usepackage{tikz,tkz-tab}
\usetikzlibrary{positioning,chains,fit,shapes,calc,arrows.meta,automata,arrows,decorations.markings,patterns}
\usepackage{tkz-graph}
\usepackage{tabularx}

\usetikzlibrary{decorations.pathreplacing,calligraphy}
\usepackage{multirow}
\usepackage{hhline}
\usepackage{url}
\usepackage{booktabs,makecell}
\usepackage[ruled, linesnumbered]{algorithm2e}
\DontPrintSemicolon
\usepackage{xcolor}

\usepackage{etoolbox}
\makeatletter
\patchcmd{\@makecaption}
  {\scshape}
  {}
  {}
  {}
\makeatletter
\patchcmd{\@makecaption}
  {\\}
  {.\ }
  {}
  {}
\makeatother

\SetKwRepeat{Do}{do}{while}
\SetKw{Break}{break}
\SetKwFor{When}{when}{do}{end}

\newtheorem{lemma}{Lemma}

\newtheorem{theorem}{Theorem}

\newtheorem{definition}{Definition}

\newcommand\opt{\textnormal{opt}\xspace}

\newcommand\ROO{\textsc{Roo}\xspace}
\newcommand\COM{\textsc{Com}\xspace}
\newcommand\AMD{\textsc{Amd}\xspace}
\newcommand\GEN{\textsc{Gen}\xspace}

\newcommand{\new}[1]{\textcolor{black}{#1}}
\newcommand{\neww}[1]{\textcolor{black}{#1}}

\usepackage[colorinlistoftodos,prependcaption,textsize=footnotesize]{todonotes}

\def\abstract{\vspace{.5em}
{\textit{\bf Abstract}:\,\relax%
}}

\def\keywords{\vspace{.5em}
{\textit{\bf Keywords}:\,\relax%
}}
\def\endkeywords{\par}

\begin{document}
\title{Improved Online Scheduling of Moldable Task Graphs under Common Speedup Models\footnote{A preliminary version \cite{Benoit22_online} of this paper appeared in the 51st International Conference on Parallel Processing (ICPP), 2022. }
}
\author{Lucas Perotin$^\S$, Hongyang Sun$^\dag$\\
\\
$^\S$Laboratoire LIP, ENS Lyon, Lyon, France\\
$^\dag$University of Kansas, Lawrence, KS, USA\\
lucas.perotin@ens-lyon.fr; hongyang.sun@ku.edu
}
\date{}
\maketitle

\begin{abstract}
We consider the online scheduling problem of moldable task graphs on multiprocessor systems for minimizing the overall completion time (or makespan). Moldable job scheduling has been widely studied in the literature, in particular
when tasks have dependencies (i.e., task graphs) or when tasks are released on-the-fly
(i.e., online). However, few studies have focused on both (i.e., online scheduling of moldable task graphs).
In this paper, we design a new online scheduling algorithm for this problem and derive constant competitive ratios under several common yet realistic
speedup models (i.e., roofline, communication, Amdahl, and a general combination). These results improve the ones we have shown in the preliminary version \cite{Benoit22_online} of this paper.
We also prove, for each speedup model, a lower bound on the competitiveness of any online list scheduling algorithm that allocates processors to a task based only on the task's parameters and not on its position in the graph. This lower bound matches exactly the competitive ratio of our algorithm for the roofline, communication and Amdahl's model, and is close to the ratio for the general model.
Finally, we provide a lower bound on the competitive ratio of any deterministic online algorithm for the arbitrary speedup model,
which is not constant but depends on the number of tasks in the longest path of the graph.

\keywords Task graph, moldable task, online scheduling, competitive ratio.
\endkeywords
\end{abstract}

\maketitle

\section{Introduction}

This work investigates the online scheduling of parallel task graphs on a set of identical processors, where each task in the graph is \emph{moldable}.
In the scheduling literature, a moldable task (or job) is a parallel task that can be executed on an arbitrary but fixed number of processors. The execution time of the task depends upon the number of processors used to execute it, which is chosen once and for all when the task starts its execution but cannot be modified later on during execution. This corresponds to a variable static resource allocation, as opposed to a fixed static allocation (\emph{rigid} tasks) and to a variable dynamic allocation (\emph{malleable} tasks)~\cite{Feitelson96}.

Moldable tasks offer a nice trade-off between rigid and and malleable tasks: they easily adapt to the number of available resources, contrarily to rigid tasks, while being easy to design and implement, contrarily to malleable tasks. Thus, many computational kernels in scientific libraries (e.g., for numerical linear algebra and tensor computations) can be deployed as moldable tasks to efficiently run on a wide range of processors. Because of the importance and wide availability of moldable tasks,
scheduling algorithms for such tasks have received considerable attention in the literature (see Section \ref{sec.related} for details). Like many other scheduling problems, scheduling moldable tasks comes in different flavors, and the following provides a brief taxonomy:

\begin{itemize}
\item \textbf{Offline Scheduling vs. Online Scheduling.} In the offline version of the scheduling problem, all tasks are known in advance, before the execution starts. The problem is $\mathcal{NP}$-complete for both independent and dependent tasks \cite{Ullman75_NPcomplete}, and the goal is to design scheduling algorithms with good \emph{approximation ratio}, which measures the worst-case performance of an algorithm against an optimal scheduler for all possible input instances. On the contrary, in the online version of the scheduling problem, tasks are released on the fly, and the objective is to design online algorithms with good \emph{competitive ratio}~\cite{Sleator1985}, against an optimal offline scheduler that knows in advance all the tasks and their dependencies in the graph. The competitive ratio is established against all possible strategies devised by an adversary trying to force the online algorithm to take \emph{bad} decisions.
\item \textbf{Independent Tasks vs. Task Graphs.} In a scheduling problem, different tasks can be either independent of each other or dependent forming a task graph. If tasks are independent, they are either all known to the scheduling algorithm initially (in the offline version), or released on the fly and the scheduler only discovers their characteristics upon release (in the online version). For task graphs, either the entire graph is known at the start (in the offline version), or each new task along with its characteristics is only released when all of its predecessors have completed execution (in the online version). For the latter case, the shape of the graph as well as the nature of the tasks are not known in advance, and they are revealed only as the execution progresses.
\end{itemize}

In this paper, we investigate arguably the most difficult version of the problem, namely, the \emph{online scheduling of moldable tasks graphs}, which is the version that has received the least research attention so far. The objective is to minimize the overall completion time of the task graph, or the \emph{makespan}. We assume that the scheduling of each task is non-preemptive and without restarts~\cite{Feldmann98_DAG}, which is highly desirable to avoid large overheads incurred by checkpointing partial results, context switching, and task migration, etc.

While the performance of a scheduling algorithm greatly depends upon the speedup model of the tasks, we consider several common yet realistic speedup models, including the roofline model, the communication model, the Amdahl's model, and a general combination of them (see Section~\ref{sec.model} for their precise definitions). These models have been widely assumed and studied in the literature for modeling the scaling behavior of parallel applications.
We extend and improve upon our preliminary work \cite{Benoit22_online} by designing a new online scheduling algorithm that achieves better competitive ratios for these models. This is done through a novel analysis framework, which  provides a tighter and more coupled analysis between a local processor allocation algorithm and the list scheduling algorithm. To the best of our knowledge, a competitive ratio was previously known only for task graphs under the roofline model~\cite{Feldmann98_DAG}, while our work offers the first competitive results for several other speedup models, and altogether these results lay the theoretical foundations for this difficult scheduling problem.

Our main contributions are summarized as follows:
\begin{itemize}
  \item We present a new online algorithm and prove its constant competitive ratio for the four speedup models. In particular, the results for the communication model, the Amdahl's model, and the general model improve upon those previously obtained in \cite{Benoit22_online}. Table~\ref{tab.all-bounds} lists our results in comparison with the ones proven in \cite{Benoit22_online}.
  \item For each speedup model, we prove a lower bound on the competitiveness of any online list scheduling algorithm whose processor allocation is \emph{local and deterministic}, i.e., the decision depends only on the total number of processors and a task's parameters but not on its position in the graph. The results show that our algorithm achieves the best possible competitive ratios for the roofline, communication and Amdhal's models for this class of algorithms.
  \item We derive a lower bound on the competitiveness of any deterministic online algorithm under the arbitrary speedup model, and show that it is not constant but depends on the number of tasks in the longest path of the graph.
\end{itemize}

\begin{table}[t]
\centering
\caption{Competitive ratios for our algorithm and the algorithm in \cite{Benoit22_online} for four speedup models. }
\label{tab.all-bounds}
\begin{tabular}{|l|c|c|c|c|}
\hline
\textbf{Speedup Model}  & \textbf{Roofline} & \textbf{Comm.} & \textbf{Amdahl} & \textbf{General} \\ \hline
Results of this paper & \new{$\approx 2.62$} & \new{$\approx 3.39$} & \new{$\approx 4.55$} & \new{$\approx 4.63$}\\
\hline
Results of \cite{Benoit22_online} & \new{$\approx 2.62$} & \new{$\approx 3.61$} & \new{$\approx 4.74$} & \new{$\approx 5.72$}\\ \hline
\end{tabular}
\end{table}

The rest of this paper is organized as follows. Section~\ref{sec.related} surveys some related works on moldable task scheduling.
The formal model and problem statement are presented in Section~\ref{sec.pb}.
Section~\ref{sec.alg} introduces the new online algorithm and proves its competitive ratios for different speedup models. Section \ref{sec.low} presents, for each model, a lower bound of any online list scheduling algorithm with deterministic local processor allocation. Our algorithm belongs to this class and has the best possible competitive ratio for the roofline, communication, and Amdahl's models.
Section~\ref{sec.lower} is devoted to proving a lower bound of any deterministic online algorithm for the arbitrary speedup model.
Finally, Section~\ref{sec.conclusion} concludes the paper and provides hints for future directions.

\section{Related Works}
\label{sec.related}

In this section, we discuss some related works on moldable task scheduling. Following the taxonomy of the previous section, we consider four versions of the problem combining offline vs.~online scheduling and independent tasks vs.~task graphs scheduling. We mainly focus on works that have derived approximation or competitive ratios. While some of these results depend on specific speedup models, others hold for a more general class of models.

\paragraph{Offline Scheduling of Independent Tasks.}
Belkhale and Banerjee \cite{Belkhale90} considered moldable tasks that follow the monotonic model, where the execution time of a task is non-increasing and the area (processor allocation times execution time) is non-decreasing with the number of processors. They presented a $\frac{2}{1+1/P}$-approximation algorithm by iteratively updating the processor allocations. For the same model, B{\l}a\.{z}ewicz et al. \cite{Blazewicz01} also presented a 2-approximation algorithm while relying on an optimal continuous schedule. Mouni\'{e} et al. \cite{Mounie99_sqrt3} presented a $(\sqrt{3}+\epsilon)$-approximation algorithm using dual approximation. They later improved the ratio to $1.5+\epsilon$ \cite{Mounie07_dualapprox}. Finally, Jansen and Land~\cite{Jansen18_PTAS} proposed a Polynomial-Time Approximation Scheme (PTAS) for the problem.

If the execution time of a task can be an arbitrary function of the processor allocation (i.e., the arbitrary model), Turek et al.~\cite{Turek92} designed a 2-approximation list-based algorithm and a 3-approximation shelf-based algorithm. Ludwig and Tiwari \cite{Ludwig94} showed the same result but with lower computational complexity. When each task only admits a subset of all possible processor allocations, Jansen \cite{Jansen12_3over2} presented a $(1.5+\epsilon)$-approximation algorithm, which is tight since the problem cannot have an approximation ratio better than 1.5 unless $\mathcal{P} = \mathcal{NP}$~\cite{Johannes06_list}. When the number of processors is a constant or polynomially bounded by the number of tasks, Jansen et al.
\cite{Jansen10_PTAS} showed that a PTAS exists.

\paragraph{Online Scheduling of Independent Tasks.}
For scheduling independent moldable tasks that arrive online over time, Havill and Mao \cite{Havill08_SET} presented a 4-competitive algorithm for the communication model by minimizing the execution time of each task. For the same model, Dutton and Mao~\cite{Dutton07_ECT} and Kell and Havill \cite{Kell15_Improved} further presented improved competitive results when the total number $P$ of processors is bounded by a constant.

Under an alternative online setting, where independent moldable tasks with the same arrival time are released one-by-one to the scheduler, Ye et al.~\cite{Ye18_online} designed a 16.74-competitive algorithm for the arbitrary speedup model.
They also showed how to transform an online algorithm for rigid tasks whose makespan is at most $\rho$ times the lower bound into a $4\rho$-competitive algorithm for moldable tasks.

\paragraph{Offline Scheduling of Task Graphs.}
For offline scheduling of moldable task graphs, Wang and Cheng~\cite{Wang92_DAG} showed that the earliest completion time algorithm is a $(3-\frac{2}{P})$-approximation for the roofline model.
Since the processor allocation is done independently for each task, their algorithm and corresponding ratio can also be applied to the online setting as discussed below.

For the monotonic model, Belkhale and Banerjee \cite{Belkhale91_DAG} presented a 2.618-approximation algorithm while assuming the availability of an optimal processor allocation.
Lep\`{e}re et al. \cite{Lepere01_DAG} proposed an algorithm with an approximation ratio of $5.236$. They also showed that the optimal allocation can be achieved in pseudo-polynomial time for some special graphs, such as series-parallel graphs and trees, thus leading to a 2.618-approximation for these graphs.
Jansen and Zhang \cite{Jansen06_DAG} later improved the approximation ratio for general graphs to around 4.73.
When assuming that the area of a job is a concave function of the number of processors, Jansen and Zhang \cite{Jansen05_concave} proposed a 3.29-approximation algorithm. Chen and Chu \cite{Chen13_concave} improved the ratio to around 2.95 by further assuming that the execution time of a job is strictly decreasing in the number of allocated processors.

\paragraph{Online Scheduling of Task Graphs.}
Feldmann et al.~\cite{Feldmann98_DAG} designed an online algorithm to schedule moldable task graphs under the roofline model. They showed that their algorithm achieves a competitive ratio of 2.618, thus improving the previous result by Wang and Cheng~\cite{Wang92_DAG}.
Furthermore, their algorithm works in the \emph{non-clairvoyant} setting, where the task execution time is also unknown to the scheduler.

Benoit et al.~\cite{Benoit20_cluster, Benoit21_ieeetc} recently investigated the problem of scheduling moldable tasks subject to failures,
where a task needs to be re-executed after a failure until it is successfully completed.
This corresponds to a special task graph consisting of multiple linear chains (one per task), where the length of each chain corresponds to the total number of executions of a task. The problem is semi-online, since all the tasks are known at the beginning, but the task failures and hence their re-executions are only discovered on-the-fly. They considered several common speedup models (as in this paper) and presented a scheduling algorithm that achieves constant competitive ratios for these models.
In this paper, we extend our preliminary work \cite{Benoit22_online} and study the general online scheduling of moldable task graphs (as in~\cite{Feldmann98_DAG}). We present improved competitive ratios as well as lower bounds for the common speedup models. We do not consider task failures as in~\cite{Benoit20_cluster, Benoit21_ieeetc}, but our results can readily carry over to the failure scenario.

Table \ref{tab.references} summarizes different versions of the moldable task scheduling problem together with the related papers within each version.

\begin{table}[t]
\centering
\caption{Difference versions of moldable task scheduling problem and related papers in each version.}
\label{tab.references}
\begin{tabular}{|l|c|c|}
\hline
 & \textbf{Offline Scheduling} & \textbf{Online Scheduling} \\ \hline
\textbf{Independent Tasks} & \parbox[c][1cm]{6cm}{\emph{Monotonic model}: \cite{Belkhale90, Blazewicz01, Mounie99_sqrt3, Mounie07_dualapprox, Jansen18_PTAS} \\ \emph{Arbitrary model}: \cite{Turek92, Ludwig94, Jansen12_3over2, Jansen18_PTAS}} & \parbox[c][1cm]{6cm}{\emph{Comm. model (over time)} \cite{Havill08_SET, Dutton07_ECT, Kell15_Improved} \\ \emph{Arbitrary model (one-by-one)}: \cite{Ye18_online}} \\ \hline
\textbf{Task Graphs} & \parbox[c][1cm]{6cm}{\emph{Roofline model}: \cite{Wang92_DAG} \\ \emph{Monotonic model}: \cite{Belkhale91_DAG, Jansen05_concave, Chen13_concave, Lepere01_DAG, Jansen06_DAG}} & \parbox[c][1cm]{6cm}{\emph{Roofline model}: \cite{Wang92_DAG, Feldmann98_DAG} \\ \emph{Common models}: \cite{Benoit20_cluster, Benoit21_ieeetc, Benoit22_online}, [this paper]} \\ \hline
\end{tabular}
\end{table}

\section{Problem Statement}
\label{sec.pb}

In this section, we formally present the online scheduling model and the objective function. We also show a simple lower bound on the optimal makespan, against which the performance of our online algorithms will be measured.

\subsection{Model and Objective}
\label{sec.model}

We consider the online scheduling of a \textbf{\emph{Directed Acyclic Graph (DAG)}} of moldable tasks on a platform with $P$ identical processors. Let $G = (V, E)$ denote the task graph, where $V = \{1, 2, \dots, n\}$ represents a set of $n$ tasks and $E \subseteq V \times V$ represents a set of precedence constraints (or dependencies) among the tasks. An edge $(i, j) \in E$ indicates that task $j$ depends on task~$i$, and therefore it cannot be executed before task~$i$ is completed. Task~$i$ is called the \textbf{\emph{predecessor}} of task~$j$, and task~$j$ is called the \textbf{\emph{successor}} of task~$i$. In this work, we do not consider the costs associated with the data transfers between dependent tasks.

The tasks are assumed to be \textbf{\emph{moldable}}, meaning that the number of processors allocated to a task can be determined by the scheduling algorithm at launch time, but once the task has started executing, its processor allocation cannot be changed. The execution time $t_j(p_j)$ of a task $j$ is a function of the number $p_j$ of processors allocated to it, and we assume that the processor allocation must be an integer between $1$ and $P$. In this paper, we focus on the following execution time function:
\begin{align}
\label{eq.exec_time}
t_j(p_j) = \frac{w_j}{\min(p_j, \bar{p}_j)} + d_j + c_j (p_j - 1) \ ,
\end{align}
where $w_j$ denotes the total parallelizable work of the task, $\bar{p}_j$ denotes the maximum degree of parallelism of the task,
$d_j$ denotes the sequential work of the task, and $c_j$ denotes the communication overhead when more than one processor is used.
The execution time function in Equation~\eqref{eq.exec_time} generalizes several speedup models commonly observed for parallel applications. In particular, it contains the following well-known models as special cases:
\begin{itemize}
\item \emph{\textbf{Roofline Model}} \cite{Williams2009} (with $d_j = 0$ and $c_j = 0$):
\begin{align}
\label{eq.model.roof}
t_j(p_j) = \frac{w_j}{\min(p_j, \bar{p}_j)} \ .
\end{align}
This model assumes that the task has a linear speedup until a maximum degree of parallelism $\bar{p}_j \le P$.
\item \emph{\textbf{Communication Model}} \cite{Havill08_SET} (with $\bar{p}_j \ge P$ and $d_j = 0$):
\begin{align}
\label{eq.model.comm}
t_j(p_j) = \frac{w_j}{p_j} + c_j (p_j - 1) \ .
\end{align}
This model assumes that the work of the task can be perfectly parallelized, but there is a communication overhead when more than one processor is allocated, which increases linearly with the number of allocated processors.
\item \emph{\textbf{Amdahl's Model}} \cite{Amdahl67} (with $\bar{p}_j \ge P$ and $c_j = 0$):
\begin{align}
\label{eq.model.amdahl}
t_j(p_j) = \frac{w_j}{p_j} + d_j \ .
\end{align}
This model assumes that the task has a perfectly parallelizable fraction with work $w_j$ and an inherently sequential fraction with work $d_j$.
\end{itemize}

From the execution time function of the task $j$, we can further define the \textbf{\emph{area}} of the task as a function
of the processor allocation as follows: $a_j(p_j) = p_j t_j(p_j)$. Intuitively, the area represents
the total amount of processor resources utilized over the entire period of task execution.

In this work, we consider the \textbf{\emph{online scheduling}} model, where a task becomes available
only when all of its predecessors have been completed. This represents a common scheduling model
for \textbf{\emph{dynamic task graphs}}, whose dependencies are only revealed upon task completions \cite{JOHNSON96_dynamic,Feldmann98_DAG,Agrawal10_dynamic,Canon20_online}.  Furthermore, when a task $j$ is available,
all of its execution time parameters (i.e., $w_j$, $\bar{p}_j$, $d_j$, $c_j$) become known
to the scheduling algorithm as well. The goal is to find a feasible schedule of the task graph that minimizes its overall completion time or \textbf{\emph{makespan}}, denoted by $T$. The performance of an online scheduling algorithm is measured by its competitive ratio: the algorithm is said to be \textbf{\emph{$c$-competitive}} if, for any task graph, its makespan~$T$ is at most $c$ times the makespan $T^{\opt}$ produced by an optimal offline scheduler, i.e., $\frac{T}{T^{\opt}} \le c$. Note that the optimal offline scheduler knows in advance all the tasks and their speedup models, as well as all dependencies in the graph. The competitive ratio is established against all possible strategies by an adversary trying to force the online algorithm to take bad decisions.

\subsection{Lower Bound on Optimal Makespan}

Given the execution time function in Equation (\ref{eq.exec_time}), let us define $s_j = \sqrt{\frac{w_j}{c_j}}$. We can then compute the maximum number of processors that should be allocated to the task as:
\begin{align}\label{eq.pjmax}
p_j^{\max} = \min\left(P, \bar{p}_j, \tilde{p}_j\right) \ ,
\end{align}
where\\[-.6cm]
\begin{align*}
\tilde{p}_j = \left\{\begin{array}{ll}
\lfloor s_j \rfloor, & \text{if } t_j(\lfloor s_j \rfloor) \le t_j(\lceil s_j \rceil)  \\
\lceil s_j \rceil, & \text{otherwise}
\end{array}\right .
\end{align*}
Indeed, allocating more than $p_j^{\max}$ processors to the task will no longer decrease its execution time while only increasing its area. Thus, we can safely assume that the processor allocation of the task should never exceed $p_j^{\max}$ by any reasonable algorithm.

Furthermore,  the task is said to satisfy the \textbf{\emph{monotonic}} property~\cite{Lepere01_DAG} if the following two conditions hold:
\begin{itemize}
\item The execution time is a \textbf{\emph{non-increasing}} function of the processor allocation, i.e., $t_j(p) \ge t_j(q)$ for all $1 \le p < q \le p_j^{\max}$;
\item The area is a \textbf{\emph{non-decreasing}} function of the processor allocation, i.e., $a_j(p) \le a_j(q)$ for all $1 \le p < q \le p_j^{\max}$.
\end{itemize}

Note that the second condition above also suggests that the task cannot achieve superlinear speedup, i.e.,
\begin{align}\label{eq.speedup_bound}
\frac{t_j(p)}{t_j(q)} \le \frac{q}{p} \text{ for all } 1 \le p < q \le p_j^{\max} \ .
\end{align}

\begin{lemma}\label{lem.monotonic}
A task $j$ with execution time function given in Equation~\eqref{eq.exec_time} is monotonic if its processor allocation is in the range $[1, p_j^{\max}]$.
\end{lemma}
\vspace{-.2cm}
\begin{proof}
When the processor allocation is in the range $[1, p_j^{\max}]$, we have $p_j \le p_j^{\max} \le \bar{p}_j$. Thus, the execution time function simplifies to $t_j(p_j) = \frac{w_j}{p_j} + d_j + c_j (p_j - 1)$.
This is a convex function whose minimum value is achieved at $\tilde{p}_j$. Since we also have $p_j \le p_j^{\max} \le \tilde{p}_j$, it shows that the execution time is a non-increasing function of $p_j$ in the range $[1, p_j^{\max}]$.

Similarly, when $p_j \le p_j^{\max} \le \bar{p}_j$, the area becomes $a_j(p_j) =  p_j t_j(p_j) = w_j + d_jp_j + c_j(p_j^2 - p_j)$,
which is non-decreasing for any $p_j \ge 1$.
\end{proof}

Based on Lemma \ref{lem.monotonic}, the minimum execution time of the task is achieved as $t_j^{\min} = t_j(p_j^{\max})$ and the minimum area of the task is achieved as $a_j^{\min} = a_j(1)$.
\neww{Further, we let $t_j^{\opt}$ and $a_j^{\opt}$ denote the execution time and area of the task under the processor allocation of an optimal schedule.}
We now define two quantities that can be used as a lower bound of the optimal makespan.

\begin{definition}
Given the processor allocations of all the tasks in an optimal schedule,
\begin{itemize}
  \item the \textbf{\emph{total area}} $A^{\opt}$ of the task graph is the sum of the areas of all the tasks in the graph, i.e., $A^{\opt} = \sum_{j=1}^{n} a_j^{\opt}$.
  \item the length $L^{\opt}(f)$ of a path\footnote{A path $f$ consists of a sequence of tasks with linear dependency, i.e., $f = (j_{\pi(1)}, j_{\pi(2)}, \dots, j_{\pi(v)})$, where the first task~$j_{\pi(1)}$ in the sequence has no predecessor in the graph, the last task $j_{\pi(v)}$ has no successor, and, for each $2\le i\le v$, task $j_{\pi(i)}$ is a successor of task~$j_{\pi(i-1)}$.} $f$ in the graph is the sum of the execution times of all the tasks along that path, i.e., $L^{\opt}(f) = \sum_{j\in f}t_j^{\opt}$. The \textbf{\emph{critical path length}} $C^{\opt}$ of the graph is the longest length of any path in the graph, i.e., $C^{\opt} = \max_{f} L^{\opt}(f)$.
\end{itemize}
\end{definition}

Clearly, the optimal makespan cannot be smaller than \new{$\frac{A^{\opt}}{P}$ and $C^{\opt}$}. This follows from the well-known area and critical-path bounds for scheduling any task graph \cite{Graham69}.
The following lemma states this result.
\begin{lemma}\label{lem.lower}
\new{$T^{\opt} \ge \max \left(\frac{A^{\opt}}{P}, C^{\opt} \right)$}.
\end{lemma}

\section{A New Online Algorithm}
\label{sec.alg}

In this section, we present a new online scheduling algorithm and derive its competitive ratio for the general speedup model (Equation~(\ref{eq.exec_time})) and its three special cases.

\subsection{Algorithm Description}

Algorithm~\ref{alg.online} presents the pseudocode of the online scheduling algorithm, which at any time maintains the set of available tasks in a waiting queue $Q$. At time 0 or whenever a running task completes execution, it checks if new tasks have become available. If so, for each newly available task $j$, it finds a processor allocation $p_j$ for the task (using Algorithm \ref{alg.allocate1}) before inserting it into the queue $Q$. Then, it applies the well-known \textbf{\emph{list scheduling}} strategy \cite{Graham69} by scanning through all the available tasks in $Q$ and executing each one right away if there are enough processors. Note that tasks are inserted into the queue without any priority considerations, although in practice certain priority rules may work better.

Algorithm \ref{alg.allocate1} presents the details of the processor allocation strategy for any task $j$. It consists of two steps. The first step performs an initial allocation for the task, which is inspired by the \textbf{\emph{local processor allocation}} strategy proposed in \cite{Benoit20_cluster, Benoit21_ieeetc}. Specifically, for each possible allocation $p\in [1, p_j^{\max}]$, we define the following:
\begin{itemize}
  \item $g_j(p) \triangleq \frac{a_j(p)}{a_j^{\min}}$: ratio between the area of the task and its minimum area;
  \item $f_j(p) \triangleq \frac{t_j(p)}{t_j^{\min}}$: ratio between the execution time of the task and its minimum execution time.
\end{itemize}

We then find an allocation $p$ that minimizes \new{$f_j(p)$} subject to the constraint \new{$g_j(p) \le \alpha^M$}, where \new{$\alpha^M \geq 1$ is a constant} whose
exact value will be determined based upon the specific speedup model $M$ under consideration. Since $g_j(p)$ is non-decreasing with $p$ and $f_j(p)$ is non-increasing with $p$, the above optimization problem can be efficiently solved \neww{using binary search in $O(\log P)$ time.}

In the second step, the algorithm reduces the initial allocation to $\lceil \mu^M P \rceil$ if it is more than $\lceil \mu^M P \rceil$; otherwise the allocation will be unchanged. \neww{Here, $\mu^M \le 0.5$ is also a constant whose value will be determined by the speedup model $M$.} Let $p_j$ denote the initial allocation for the task and $p'_j$ the final allocation. Thus, after the second step, we have:
\begin{align}\label{eq.adjust}
p'_j = \begin{cases}
\lceil \mu^M P \rceil,  & \text{~if~} {p}_j > \lceil \mu^M P \rceil \\
{p}_j, & \text{~otherwise~}
\end{cases} \ .
\end{align}

This step adopts the technique first proposed in \cite{Lepere01_DAG} and subsequently used in \cite{Jansen06_DAG, Jansen05_concave}. The purpose is to enable the execution of more tasks at any time, thus potentially increasing the overall resource utilization of the platform and reducing the makespan.

\neww{The exact values of the two parameters $\alpha^M$ and $\mu^M$ for each speedup model $M$ will be presented in Section \ref{sec.comp-ratio} when analyzing the above algorithm. }

{
\begin{algorithm}[t]
    \footnotesize
	\caption{Online\_Scheduling\_Algorithm}\label{alg.online}
        initialize a waiting queue $Q$\;
        \When{\textnormal{at time $0$ or a running task completes execution}}{
            \tcp{Processor Allocation}
            \For{\textnormal{each new task $j$ that becomes available}}{
                Allocate\_Processor($j$)\;
                insert task $j$ into the waiting queue $Q$\;
            }
            \tcp{List Scheduling}
            \For{\textnormal{each task $j$ in the waiting queue $Q$}}{
                \If{\textnormal{there are enough processors to execute the task}}{
                    execute task $j$ now\;
                }
            }
		}
\end{algorithm}
}
{
\begin{algorithm}[t]
    \footnotesize
	\caption{Allocate\_Processor($j$)}\label{alg.allocate1}
    \KwIn{\neww{Task $j$ and its speedup model $M$; parameters $\alpha^M, \mu^M$}}
    \KwOut{\neww{Processor allocation $p'_j$ for the task}}
    \tcp{Step 1: Initial Allocation}
    Compute $p_j^{\max}$ based on Equation (\ref{eq.pjmax})\;
    Compute $t_j^{\min} = t_j(p_j^{\max})$ and $a_j^{\min} = a_j(1)$\;
    Find an allocation $p_j \in [1, p_j^{\max}]$ from the following optimization problem:\new{
    \begin{align*}
    \min_{p}~f_j(p) &\triangleq \frac{t_j(p)}{t_j^{\min}} \\
    \text{s.t.}~g_j(p) &\triangleq \frac{a_j(p)}{a_j^{\min}} \le \alpha^M
    \end{align*} }\\
    \tcp{Step 2: Allocation Adjustment}
    \lIf{$p_j > \lceil \mu^M P \rceil$}
    {
        $p'_j \leftarrow \lceil \mu^M P \rceil$ \textbf{else} $p'_j \leftarrow p_j$
    }
\end{algorithm}
}

\subsection{A Novel Analysis Framework}\label{sec.framework}

We first outline an analysis framework, under which the competitive ratio of the proposed online algorithm will be derived for different speedup models.
The framework \new{is inspired by} the analysis shown in \cite{Lepere01_DAG, Jansen06_DAG, Jansen05_concave} for list scheduling as well as the analysis used in \cite{Benoit20_cluster, Benoit21_ieeetc} for local processor allocation. Together, the result nicely connects the makespan of the online algorithm to the lower bound (Lemma \ref{lem.lower}), thus proving the competitive ratio. We point out that, compared to the analysis in our preliminary work \cite{Benoit22_online}, the novel framework presented in this section provides a tighter and more coupled analysis of both list scheduling and local processor allocation, which ultimately leads to the improved competitive ratios.

Since the analysis framework in this section applies to any speedup model $M$, for simplicity, we will drop the superscript $M$ for $\alpha^M$ and $\mu^M$, and re-introduce it later in Section \ref{sec.comp-ratio} for each specific speedup model we consider.

Recall that $T$ denotes the makespan of the online scheduling algorithm. Since the algorithm allocates and de-allocates processors upon task completions, the schedule can be divided into a set $\mathcal{I} = \{I_1, I_2, \dots\}$ of non-overlapping intervals, where tasks only start (or complete) at the beginning (or end) of an interval, and the number of utilized processors does not change during an interval. For each interval $I \in \mathcal{I}$, let $p(I)$ denote its processor utilization, i.e., the total number of processors used by all tasks running in interval $I$. \new{We first classify the set of all intervals into the following two categories:
\begin{itemize}
\item $\mathcal{I}_0$: subset of intervals that satisfy $p(I) \in (0, \lceil (1-\mu) P \rceil)$;
\item $\mathcal{I}_3$: subset of intervals that satisfy $p(I) \in [\lceil (1-\mu) P \rceil, P]$.
\end{itemize}}

The following lemma shows a property for the subset of intervals in $\mathcal{I}_0$.
\begin{lemma}\label{lem.I0_path}
\new{There exists a path $f$ in the graph in which a task is always running in $\mathcal{I}_0$.}
\end{lemma}

\begin{proof}
During \new{$\mathcal{I}_0$}, the processor utilization is at most $\lceil (1-\mu) P \rceil - 1$, so there are at least $P - (\lceil (1-\mu) P \rceil - 1) \ge \lceil \mu P \rceil$ available processors. Based on Algorithm \ref{alg.allocate1}, any task is allocated at most $\lceil \mu P \rceil$ processors. Thus, there are enough processors to execute any new task (if one is available). This implies that there is no available task in the queue $\mathcal{Q}$ during \new{$\mathcal{I}_0$}. When a task graph is scheduled by the list scheduling algorithm, it is well known that there exists a path $f$ in the graph such that some task along that path will be running whenever there is no available task in the queue \cite{Feldmann98_DAG, Lepere01_DAG, Jansen06_DAG}, \new{hence the result}.
\end{proof}

\new{
Using the path $f$ stated in Lemma \ref{lem.I0_path}, we further split $\mathcal{I}_0$ into the following two sub-categories:
\begin{itemize}
\item $\mathcal{I}_1$: subset of $\mathcal{I}_0$ where the processor allocation for the currently running task in $f$ was not reduced (by the second step of Algorithm \ref{alg.allocate1});
\item $\mathcal{I}_2$:  subset of $\mathcal{I}_0$ where the processor allocation for the currently running task in $f$ was reduced (i.e., the task is running on $\lceil \mu P \rceil$ processors).
\end{itemize}
}

\new{Finally, given the processor allocation of an optimal schedule, we further split $\mathcal{I}_2$ into the following two sub-categories:
\begin{itemize}
\item \neww{$\mathcal{I}_{2'}$}: subset of $\mathcal{I}_2$ where the currently running task in $f$ was allocated with strictly fewer processors than in the optimal schedule;
\item \neww{$\mathcal{I}_{2''}$}: subset of $\mathcal{I}_2$ where the currently running task in $f$ was allocated with \neww{equal or more} processors than in the optimal schedule.
\end{itemize}
}

Let $|I|$ denote the duration of an interval $I$, and let $T_1 = \sum_{I\in \mathcal{I}_1} |I|$, $T_2 = \sum_{I\in \mathcal{I}_2} |I|$, \new{$T_{2'} = \sum_{I\in \mathcal{I}_{2'}} |I|$, $T_{2''} = \sum_{I\in \mathcal{I}_{2''}} |I|$} and $T_3 = \sum_{I\in \mathcal{I}_3} |I|$ denote the total durations of the different categories of intervals, respectively. Since $\mathcal{I}_1$, $\mathcal{I}_2$ and $\mathcal{I}_3$ are obviously disjoint and partition $\mathcal{I}$, we have $T = T_1 + T_2 + T_3$. \new{Finally, we define $z \in [0,1]$ such that $T_{2'}=zT_2$ and $T_{2''}=(1-z)T_2$}.

The next two lemmas relate these durations to the total area $A_{\opt}$ and critical path length $C_{\opt}$ of the task graph under an optimal schedule, given certain conditions on the initial processor allocations of the tasks under our algorithm.
\begin{lemma}\label{lem.area}
If there exists a constant $\alpha$ such that, for each task $j$, its initial processor allocation satisfies $a_j(p_j) \le \alpha a_j^{\min}$, then we have:
\begin{equation}\label{Area}
\new{\mu \left( z+\frac{1-z}{\alpha} \right)T_2 + \frac{(1-\mu)}{\alpha} T_3 \le T^{\opt} \ . } 
\end{equation}
\end{lemma}

\begin{proof}
\new{As the area of each task $j$ is non-decreasing with its processor allocation and $p'_j \le p_j$, the final area of the task should satisfy $a_j(p'_j) \le a_j(p_j) \le \alpha a_j^{\min} \leq \alpha a_j^{\opt}$. Furthermore, during $\mathcal{I}_{2'}$, any running task $j$ from path $f$ satisfies $a_j(p'_j) \leq a_j(p_j^{\opt})=a_j^{\opt}$. We let $A_{2'|f}$ (resp. $A_{2''|f}$) denote the total area of the fraction of tasks from $f$ running in $\mathcal{I}_{2'}$ (resp. $\mathcal{I}_{2''}$), and $A^{\opt}_{2'|f}$ (resp. $A^{\opt}_{2''|f}$) the corresponding fraction of area in an optimal schedule. We have $A_{2'|f} \leq A^{\opt}_{2'|f}$ and $A_{2''|f} \leq \alpha A^{\opt}_{2''|f}$. Since $\lceil \mu P \rceil \geq \mu P$ processors are used to run tasks from $f$ in $\mathcal{I}_{2'} \cup \mathcal{I}_{2''}$, we have $\mu T_{2'} \leq \frac{{A_{2'}}_{|f}}{P} \le \frac{A^{\opt}_{2'|f}}{P}$ and $\mu T_{2''} \leq \frac{{A_{2''}}_{|f}}{P} \le \frac{\alpha A^{\opt}_{2''|f}}{P}$.}

\new{Finally, let $A_3$ denote the total area of the fraction of tasks running in $\mathcal{I}_{3}$ and $A_3^{\opt}$ the corresponding fraction of area in an optimal schedule. Since at least $\lceil (1-\mu)P \rceil \ge (1-\mu)P$ processors are utilized during $\mathcal{I}_{3}$, we have $(1-\mu) T_3 \le \frac{A_3}{P} \le \frac{\alpha A_3^{\opt}}{P}$.}

\new{Thus, altogether we can derive:
\begin{align*}
& \mu \left(z+\frac{1-z}{\alpha} \right)T_2 + \frac{(1-\mu)}{\alpha} T_3 \\
&= \mu T_{2'} + \frac{\mu T_{2''}}{\alpha} + \frac{(1-\mu)}{\alpha} T_3\\
& \leq  \frac{A^{\opt}_{2'|f}}{P} + \frac{A^{\opt}_{2''|f}}{P}+\frac{A^{\opt}_3}{P} \\
& \leq \frac{A^{\opt}}{P} \le T^{\opt}\ .
\end{align*}
The last inequality is due to the makespan lower bound shown in Lemma \ref{lem.lower}.
}
\end{proof}

\begin{lemma}\label{lem.path}
If there exists a constant $\beta$ such that, for each task~$j$, its initial processor allocation satisfies $t_j(p_j) \le \beta t_j^{\min}$, then we have:
\begin{equation}\label{CP}
\new{\frac{T_1}{\beta}  + (\mu z+1-z) T_2 \le T^{\opt} \ .}
\end{equation}
\end{lemma}

\begin{proof}
For any task $j$ from path $f$ running during $\mathcal{I}_1$, its processor allocation was not reduced by the second step of Algorithm \ref{alg.allocate1}, thus we must have $p'_j = p_j \le \lceil \mu P \rceil$. Therefore, its execution time should satisfy $t_j(p'_j) = t_j(p_j) \le \beta t_j^{\min} \le \beta t_j^{\opt}$.

For any task $j$ from path $f$ running during $\mathcal{I}_{2'}$, its processor allocation has been reduced, i.e., $p'_j = \lceil \mu P \rceil$. Based on
Equation~\eqref{eq.speedup_bound}, the task's execution time should satisfy $\frac{t_j(p'_j)}{t_j^{\min}} = \frac{t_j(\lceil \mu P \rceil)}{t_j(p_j^{\max})} \le \frac{p_j^{\max}}{\lceil \mu P \rceil} \le \frac{P}{\mu P} = \frac{1}{\mu}$. Thus, we have $t_j(p'_j) \le \frac{1}{\mu} t_j^{\min} \le \frac{1}{\mu} t_j^{\opt}$.

\new{Finally, for any task $j$ from path $f$ running during $\mathcal{I}_{2''}$, its processor allocation is higher than that of an optimal schedule. Therefore, its execution time should satisfy: $t_j(p'_j) \le t_j(p_j^{\opt}) = t_j^{\opt}$.}

\neww{Now, let $L^{\opt}_{1|f}$ (resp. $L^{\opt}_{2'|f}$ and $L^{\opt}_{2''|f}$) denote the length for the portion of path $f$ executed
during $\mathcal{I}_1$ (resp. $\mathcal{I}_{2'}$ and $\mathcal{I}_{2''}$) under an optimal schedule. The argument above implies that $T_1 \le \beta L^{\opt}_{1|f}$, $T_{2'} \le \frac{1}{\mu} L^{\opt}_{2'|f}$ and $T_{2''} \le L^{\opt}_{2''|f}$.
Thus, we can derive:
\begin{align*}
& \frac{T_1}{\beta}  + (\mu z+1-z) T_2 \\
&= \frac{T_1}{\beta} + \mu T_{2'}+ T_{2''} \\
&\le L^{\opt}_{1|f} + L^{\opt}_{2'|f} + L^{\opt}_{2''|f} \\
&\le L^{\opt}(f) \le C^{\opt} \le T^{\opt} \ .
\end{align*}
The last inequality is again due to the makespan lower bound shown in Lemma \ref{lem.lower}.
}
\end{proof}

Based on the results of Lemmas \ref{lem.area} and \ref{lem.path}, we can now derive an upper bound on the makespan of the online scheduling algorithm as shown below.

\begin{lemma}\label{lem.makespan}
\neww{If there exist two constants $\alpha$ and $\beta$ such that, for each task $j$, its initial processor allocation satisfies:
\begin{align}
g_j(p_j) \triangleq \frac{a_j(p_j)}{a_j^{\min}} &\le \alpha \ , \label{eq.arebound} \\
f_j(p_j) \triangleq \frac{t_j(p_j)}{t_j^{\min}} &\le \beta \ ,  \label{eq.timebound}
\end{align}
then by setting $\mu$ such that $\beta+\frac{\alpha}{1-\mu}=\frac{1}{\mu}$, i.e., $\mu=\frac{\alpha+\beta+1-\sqrt{(\alpha+\beta+1)^2-4\beta}}{2\beta}$, and under the condition $\beta \ge \frac{\mu(\alpha-1)}{(1-\mu)^2}$, we get:
\begin{equation}\label{eq.makespan}
\frac{T}{T^{\opt}} \leq \frac{1}{\mu}=\frac{2\beta}{\alpha+\beta+1-\sqrt{(\alpha+\beta+1)^2-4\beta}} \ .
\end{equation}
}
\end{lemma}

\begin{proof}
As the makespan is given by $T=T_1+T_2+T_3$, we can multiply both sides by $\frac{1-\mu}{\alpha}$ and apply
Equation~\eqref{Area} to remove the $T_3$ term, which gives:
$$\frac{1-\mu}{\alpha}T \leq \frac{1-\mu}{\alpha} T_1+\new{\left(\frac{1-\mu-z\alpha \mu-(1-z)\mu}{\alpha} \right)}T_2+T^{\opt} \ . $$

We can then multiply both sides of the above inequality by \new{$\frac{\alpha}{(1-\mu)\beta}$} and apply Equation~\eqref{CP} to remove the $T_1$ term. This gives:
\new{$$\frac{T}{\beta} \leq \left(\frac{1-\mu-z\alpha \mu-(1-z)\mu}{(1-\mu)\beta}-\mu z+z-1\right)T_2+\left(1+\frac{\alpha}{(1-\mu)\beta}\right)T^{\opt} \ . $$}

\new{Now, if $f(z)=\frac{1-\mu-z\alpha \mu-(1-z)\mu}{(1-\mu)\beta}-\mu z+z-1 \leq 0$ is true for all $z\in [0, 1]$, we can omit the $T_2$ term in the above inequality and get:
\begin{align*}
T\leq \left(\beta+\frac{\alpha}{(1-\mu)}\right)T^{\opt} = \frac{1}{\mu} T^{\opt} \ .
\end{align*}
}

\new{We have $f'(z)=\frac{\mu(1-\alpha)}{(1-\mu)\beta}+(1-\mu)\ge0$ under the condition $\beta \ge \frac{\mu(\alpha-1)}{(1-\mu)^2}$, which makes $f(z)$ an increasing function of $z$. Thus, we simply need to ensure that $f(1) =\frac{1-\mu-\alpha \mu}{(1-\mu)\beta}-\mu \leq 0$, which is true if $\beta+\frac{\alpha}{1-\mu}=\frac{1}{\mu}$. One can then solve for $\mu$ from the above second-degree equation, and get $\mu=\frac{\alpha+\beta+1-\sqrt{(\alpha+\beta+1)^2-4\beta}}{2\beta}$.}

\neww{Finally, we show that the value of $\mu$ above is well-defined and is a valid choice satisfying $\mu \in (0, 0.5]$.
\begin{itemize}
  \item First, we can derive that $\Delta =(\alpha + \beta +1)^2-4\beta > (\beta+1)^2-4\beta = (\beta-1)^2 \ge 0$. Thus, the value of $\mu$ is well-defined.
  \item We have $\mu>\frac{\alpha + \beta +1-\sqrt{(\alpha + \beta +1)^2}}{2\beta}=0$, since $\alpha, \beta>0$.
  \item We can show $\mu=\frac{\alpha+\beta+1-\sqrt{(\alpha+\beta+1)^2-4\beta}}{2\beta} \le 0.5$, which after some manipulations is equivalent to showing $0 \le \beta^2 + 2\beta(\alpha-1)$. The latter inequality is always true since $\alpha \ge 1$. \qedhere
\end{itemize}
}
\end{proof}

\neww{\emph{Remarks.}~We point out that the two bounds shown in Lemma \ref{lem.makespan}, i.e., Inequalities (\ref{eq.arebound}) and (\ref{eq.timebound}), must be satisfied by all tasks in a task graph for the derived competitive ratio shown in Equation (\ref{eq.makespan}) to hold. In Section \ref{sec.comp-ratio}, we will prove that, for a given speedup model and regardless of the task parameters, there always exists a processor allocation that satisfies the two bounds with a particular $(\alpha, \beta)$ choice that achieves the minimum (or close to minimum) competitive ratio. We then show in Section \ref{sec.low} that the obtained competitive ratios are tight (or almost tight) by proving matching (or nearly matching) lower bounds for different speedup models.}
\neww{Thus, given an $(\alpha, \beta)$ pair for a speedup model, the processor allocation algorithm should find an allocation for each task to satisfy the two bounds. However, there could be multiple allocations that all satisfy the bounds. Among these allocations, Algorithm \ref{alg.allocate1} finds one that, subject to the $\alpha$ bound, minimizes the execution time of the task, thus satisfying the $\beta$ bound as well. Intuitively, this is a good practical choice and it also helps to simplify the analysis, which we will present in Section \ref{sec.comp-ratio}.}

\subsection{Competitive Ratios}
\label{sec.comp-ratio}

In this section, we prove the competitive ratios for the online algorithm under different speedup models.
Based on Lemma \ref{lem.makespan}, the competitive ratio is given by: $\frac{1}{\mu}=\frac{2\beta}{\alpha+\beta+1-\sqrt{(\alpha+\beta+1)^2-4\beta}}$ subject to the constraint $\beta \ge \frac{\alpha-1}{(1-\mu)^2}$.

\neww{For a given speedup model with parameters $\mathcal{P} \subseteq \{w,d,c, \bar{p} \}$, a generic approach for minimizing the ratio above can be outlined as follows: First, compute $\beta(\alpha)$ as small as possible based on the model parameters $\mathcal{P}$ for any fixed $\alpha \ge 1$. To do that, since the area of a task is non-decreasing with the processor allocation and the time non-increasing in the range $[1, p^{\max}]$, we can find the largest processor allocation $p^*(\mathcal{P}) \in [1, p^{\max}]$ that satisfies $\frac{a(p^*(\mathcal{P}))}{a^{\min}} \le \alpha$ and then compute $\beta(\alpha)=sup_{\mathcal{P}}\left(\frac{t(p^*(\mathcal{P}))}{t^{\min}} \right)$. Finally, plug $\beta(\alpha)$ into the expression of the competitive ratio and find the $\alpha$ that minimizes it while satisfying the constraint. }

\neww{Although the technique outlined above is a good generic approach, the computations involved are often too complicated for some speedup models and solving it will rely on numerical tools. Therefore, to derive the competitive ratios analytically, we will simply find a valid pair $(\alpha,\beta)$ below for each considered speedup model while verifying that the constraint is satisfied. We will then show the tightness (or near tightness) of the obtained competitive ratios by computing the lower bounds in Section~\ref{sec.low}.}

\neww{In the following, we will first consider the three special speedup models (i.e., roofline, communication and Amdahl) before tackling the general model. For clarity, we now re-introduce superscript $M \in \{\ROO, \COM, \AMD, \GEN\}$ to the notations $\alpha^M$, $\beta^M$ and $\mu^M$ in the lemmas and theorems below. Given a speedup model $M$, the analysis focuses on finding $\alpha^M$ and $\beta^M$ for each individual task, thus we will drop the task index $j$ for simplicity.}

\subsubsection{Roofline Model}
Recall that a task follows the roofline speedup model if its execution time satisfies $t(p) = \frac{w}{\min(p, \bar{p})}$ for some $\bar{p} \le P$.

\begin{lemma}\label{lem.roofline}
For any task that follows the roofline speedup model, there exists a processor allocation that achieves $\alpha^{\ROO}=\beta^{\ROO}=1$.
\end{lemma}

\begin{proof}
For any task with $\bar{p}$, setting the processor allocation to $p = \bar{p}$ clearly achieves the minimum execution time $t^{\min} = \frac{w}{\bar{p}}$ for the task. It also achieves the minimum area $a^{\min} = w$, which is not affected by the processor allocation in $[1, \bar{p}]$ due to the task's linear speedup in this range.
Thus, this gives $\alpha^{\ROO} = \beta^{\ROO} = 1$.
\end{proof}

\begin{theorem}\label{thm.roofline}
Algorithm \ref{alg.online} is \new{$\frac{2}{3-\sqrt{5}} < 2.62$}-competitive for any graph of tasks that follow the roofline speedup model. This is achieved with $\mu^{\ROO}=\frac{3-\sqrt{5}}{2} \approx 0.382$.
\end{theorem}

\begin{proof}
\new{With $\alpha=\beta=1$, the constraint $\beta \ge \frac{\mu (\alpha-1)}{(1-\mu)^2}=0$ is obviously satisfied. Thus, we get $\mu=\frac{\alpha+\beta+1-\sqrt{(\alpha+\beta+1)^2-4\beta}}{2\beta}=\frac{3-\sqrt{5}}{2}$, and the competitive ratio is given by $\frac{1}{\mu}= \frac{2}{3-\sqrt{5}} < 2.62$.}
\end{proof}

\emph{Remarks.}~The above ratio retains the same result by Feldmann et al. \cite{Feldmann98_DAG}\footnote{In~\cite{Feldmann98_DAG}, each task
has a parallelism $p$, and can be virtualized if $p' \leq p$ processors are used for execution, with a linear slowdown. This is equivalent
to the roofline model.}, who also proved a matching lower bound for any online deterministic algorithm under the ``non-clairvoyant" setting, where the work~$w$ of a task is unknown to the scheduler. In Section \ref{sec.low}, we will prove the same lower bound, but without the non-clairvoyant setting for a class of list scheduling algorithms with deterministic local decisions for processor allocation.

\subsubsection{Communication Model}
Recall that a task follows the communication model if its execution time satisfies $t(p) = \frac{w}{p} + c (p - 1)$,
with \new{$c \geq 0$}. If $c=0$, it simplifies to a special case of the roofline model \new{and we can reach $\alpha=\beta=1$}. Thus, we \new{assume $c>0$ and} rewrite the execution time function as: $t(p) = c(\frac{w'}{p}+p-1)$ with $w' = \frac{w}{c}$. \neww{The area function is then given by $a(p) = c(w'+p(p-1))$.}
\begin{lemma}\label{lem.comm}
For any task that follows the communication model, there exists a processor allocation that achieves \new{$\alpha^{\COM}=\frac{4}{3}$ and $\beta^{\COM}=\frac{3}{2}$}.
\end{lemma}

\begin{proof}
Recall that $p^{\max}$ denotes the number of processors that minimizes the task's execution time $t(p)$, i.e., $t(p^{\max}) = t^{\min}$. Clearly, we have either $p^{\max}=P$ or $\lfloor \sqrt{w'} \rfloor \leq p^{\max} \leq \lceil \sqrt{w'} \rceil $. Also, the minimum area of the task is obtained with one processor, i.e., $a^{\min} = a(1) = cw'$.

\new{Furthermore, for a given choice of \neww{$p$}, we can show that $f(w', p) \triangleq \frac{t(p)}{t^{\min}} = \frac{\frac{w'}{p}+p-1}{\frac{w'}{p^{\max}}+p^{\max}-1}$ is a non-decreasing function of $w'$ in the interval $[p^2, \infty)$. To see that, we can compute the partial derivative:
\begin{align*}
\frac{\partial f(w', p)}{\partial w'}&=\frac{\frac{1}{p}\left(\frac{w'}{p^{\max}}+p^{\max}-1\right)-\frac{1}{p^{\max}}\left(\frac{w'}{p}+p-1\right)}{\left(\frac{w'}{p^{\max}}+p^{\max}-1\right)^2}\\
&=\frac{\frac{p^{\max}-1}{p}-\frac{p-1}{p^{\max}}}{\left(\frac{w'}{p^{\max}}+p^{\max}-1\right)^2} \ ,
\end{align*}
which is defined everywhere except at the points where $p^{\max}$ changes (due to changes of $w'$), and has the same sign as $\frac{p^{\max}-1}{p}-\frac{p-1}{p^{\max}}\geq 0$. The last inequality is because if $p=p^{\max}$, it is equal to $0$, otherwise $\frac{p^{\max}-1}{p} \geq 1$ and $\frac{p-1}{p^{\max}}<1$. This remains true if $p \le p^{\max}$, which is satisfied when $w'\ge p^2$.}

We now consider three cases:

\new{\textbf{Case 1}: $w'\leq 6$. In this case, we set $p=1$, which gives the minimum area, i.e., $\frac{a(p)}{a^{\min}} = 1$. When $w' \leq 1$, setting $p=1$ also gives the minimum execution time, i.e., $\frac{t(p)}{t^{\min}} =1$. Otherwise, we have $\frac{t(p)}{t^{\min}}=f(w',1) \leq f(6,1)=\frac{6}{\min(\frac{6}{2}+1,\frac{6}{3}+2)} =\frac{3}{2}$, since $p^{\max}=2$ or $p^{\max}=3$ when $w'=6$.}

\new{\textbf{Case 2}: $6 < w' \leq 25$. In this case, we set $p=2$ and get $\frac{a(p)}{a^{\min}} =\frac{c(w'+2)}{cw'}=1+\frac{2}{w'} < \frac{4}{3}$. As $w'>p^2=4$, we can also get $\frac{t(p)}{t^{\min}}=f(w',2) \leq f(25,2)=\frac{\frac{25}{2}+1}{\frac{25}{5}+4}=\frac{27}{18}=\frac{3}{2}$, since $p^{\max}=5$ when $w'=25$.}

\new{\textbf{Case 3}: $w'>25$. In this case, we have $t^{\min} \geq c(2\sqrt{w'}-1)$, which is the minimum possible execution time if the processor allocation could be non-integers. We set $p=\left\lfloor \sqrt{\frac{w'}{3}}+\frac{1}{2} \right\rfloor$ and obtain $\frac{a(p)}{a^{\min}}=\frac{c(w'+p(p-1))}{cw'}\leq 1+\frac{1}{w'}\left( \sqrt{\frac{w'}{3}}+\frac{1}{2}\right)\left(\sqrt{\frac{w'}{3}}-\frac{1}{2}\right)\leq 1+\frac{1}{w'}\frac{w'}{3} = \frac{4}{3}$.
Finally, $\frac{t(p)}{t^{\min}} \leq \frac{c\left(\frac{w'}{\sqrt{\frac{w'}{3}}-\frac{1}{2}} +\sqrt{\frac{w'}{3}}\right)}{c(2\sqrt{w'}-1)}=\frac{1}{2-\frac{1}{\sqrt{w'}}}\left(\frac{1}{\frac{1}{\sqrt{3}}-\frac{1}{2\sqrt{w'}}} +\frac{1}{\sqrt{3}}\right)$.
This function is clearly decreasing with $w'$, and using $w'>25$, we get $\frac{t(p)}{t^{\min}} \leq \frac{5}{9}\left(\frac{10\sqrt{3}}{10-\sqrt{3}}+\frac{1}{\sqrt{3}} \right) \approx 1.48 <\frac{3}{2}$.}
\end{proof}

\new{\emph{Remarks.} Using the generic approach outlined at the beginning of this section, we could actually compute the best possible $\beta(\alpha)$ for any $\alpha>1$.
The analysis, however, is very technical, and finding the optimal $\alpha$ then requires numerical analysis tools. It turns out that the best $(\alpha, \beta)$ pair is $(\frac{4}{3}, \frac{3}{2})$, and we will show that this pair is optimal in Section~\ref{sec.low}.}

\begin{theorem}
Algorithm \ref{alg.online} is \new{$\frac{18}{23-\sqrt{313}} < 3.391$}-competitive for any graph of tasks that follow the communication model. This is achieved with \new{$\mu^{\COM}=\frac{23-\sqrt{313}}{18} \approx 0.295$}.
\end{theorem}

\begin{proof}
\new{With $\alpha=\frac{4}{3}$ and $\beta=\frac{3}{2}$, we get $\mu=\frac{\alpha+\beta+1-\sqrt{(\alpha+\beta+1)^2-4\beta}}{2\beta}=\frac{23-\sqrt{313}}{18}$, and the constraint $\beta \ge \frac{\mu(\alpha-1)}{(1-\mu)^2} \approx 0.2$ is satisfied. Thus, the competitive ratio is given by $\frac{1}{\mu}= \frac{18}{23-\sqrt{313}} < 3.391$.}
\end{proof}

\subsubsection{Amdahl's Model}

Recall that a task follows the Amdahl's model if its execution time function is $t(p) = \frac{w}{p} + d$, with \new{$d\geq0$}, thus the area function is given by $a(p)= p t(p) = w + dp$.

\begin{lemma}\label{lem.amdahl}
For any task that follows the Amdahl's model, there exists a processor allocation that achieves \neww{$\alpha^{\AMD} =\frac{\sqrt{2}+1+\sqrt{2\sqrt{2}-1}}{2} \approx 1.883$ and $\beta^{\AMD}=\frac{1+\sqrt{4\sqrt{2}+5}}{2} \approx 2.132$}.
\end{lemma}

\begin{proof}
The minimum execution time of the task is obtained with all $P$ processors, i.e., $t^{\min} = t(P) = \frac{w}{P} + d$, and the minimum area with just one processor, i.e., $a^{\min} = a(1) = w+d$.

\new{We first show that, for any $\alpha > 1$, there exists a processor allocation $p$ that satisfies the $\alpha$ bound, i.e., $\frac{a(p)}{a^{\min}} \le \alpha$ and at the same time achieves $\beta(\alpha)=\frac{\alpha}{\alpha-1}$, i.e., $\frac{t(p)}{t^{\min}} \le \beta(\alpha)=\frac{\alpha}{\alpha-1}$. To show that, for any given $\alpha>1$, we let $x = \alpha-1$, and set $p=\min(\lceil x\frac{w}{d} \rceil,P)$. This implies $p \le \lceil x\frac{w}{d} \rceil \leq x\frac{w}{d}+1$. Thus, we have
$\frac{a(p)}{a^{\min}} = \frac{w+dp}{w+d} \leq \frac{w+d\left(x\frac{w}{d}+1\right)}{w+d} = \frac{w+d+xw}{w+d} = 1 + \frac{xw}{w+d} \le 1+x = \alpha$.
Furthermore, if $p = \lceil x\frac{w}{d} \rceil \ge x\frac{w}{d}$, we have $\frac{t(p)}{t^{\min}} \leq \frac{\frac{w}{x\frac{w}{d}}+d}{\frac{w}{P}+d} \le \frac{\frac{d}{x}+d}{d} = \frac{1}{x}+1 = \frac{1}{\alpha-1}+1=\frac{\alpha}{\alpha-1}=\beta(\alpha)$. Otherwise, if $p=P$, we get $t(p) = t^{\min}$ and thus $\frac{t(p)}{t^{\min}}=1<\frac{\alpha}{\alpha-1} = \beta(\alpha)$.}

\neww{We can now substitute $\beta(\alpha)$ into the expression of the competitive ratio and get:
\begin{align*}
\frac{1}{\mu} = \frac{2\frac{\alpha}{\alpha-1}}{\alpha+ \frac{\alpha}{\alpha-1}+1-\sqrt{\left(\alpha+ \frac{\alpha}{\alpha-1}+1\right)^2-4 \frac{\alpha}{\alpha-1}}} \ .
\end{align*}
To minimize the ratio above, one can use the standard technique of differentiating and setting the derivative to zero. The expression is quite long, and the full analysis is omitted. It turns out that $\alpha^{\AMD} =\frac{\sqrt{2}+1+\sqrt{2\sqrt{2}-1}}{2}$ minimizes the ratio. Plugging it back into $\beta(\alpha) = \frac{\alpha}{\alpha-1}$ and simplifying, we can get $\beta^{\AMD} = \frac{1+\sqrt{4\sqrt{2}+5}}{2}$.
}
\end{proof}

\begin{theorem}\label{thm.amdahl}
Algorithm \ref{alg.online} is \new{$\frac{2}{1-\sqrt{8\sqrt{2}-11}} < 4.55$}-competitive for any graph of tasks that follow the Amdahl's model. This is achieved with \new{$\mu^{\AMD}=\frac{1-\sqrt{8\sqrt{2}-11}}{2} \approx 0.22$}.
\end{theorem}

\begin{proof}
\neww{By substituting $\alpha =\frac{\sqrt{2}+1+\sqrt{2\sqrt{2}-1}}{2}$ and $\beta=\frac{1+\sqrt{4\sqrt{2}+5}}{2}$ into the expression of $\mu$ and simpifying (hard!), we can get $\mu = \frac{\alpha+\beta+1-\sqrt{(\alpha+\beta+1)^2-4\beta}}{2\beta} = \frac{1-\sqrt{8\sqrt{2}-11}}{2}$. We can also check that the constraint $\beta > \frac{\mu (\alpha-1)}{(1-\mu)^2} \approx 0.32$ is satisfied.
Thus, the competitive ratio is given by $\frac{1}{\mu} = \frac{2}{1-\sqrt{8\sqrt{2}-11}} < 4.55$. }
\end{proof}

\subsubsection{General Model}

We finally consider the general speedup model as given in Equation (\ref{eq.exec_time}). \neww{Without loss of generality}, \new{we assume $w>0$ (otherwise we get $\alpha=\beta=1$ using one processor), and $c, d>0$ (otherwise the model reduces to the communication or the Amdahl's model and also results in smaller $\alpha$ and $\beta$)}. We rewrite the execution time function as: $t(p) = c\left(\frac{w'}{\min(p, \bar{p})}+d'+p-1\right)$ with $w' = \frac{w}{c}$ and $d' = \frac{d}{c}$. \neww{We further assume $\bar{p} \le P$ (otherwise changing $\bar{p}$ to $P$ does not affect the execution time of the task for any feasible processor allocation). Finally, any reasonable scheduling algorithm will not allocate more than $\bar{p}$ processors to the task since it would increase both execution time and area. Thus, assuming $p \le \bar{p}$, we can simplify the execution time function as $t(p) = c\left(\frac{w'}{p}+d'+p-1\right)$ and the area function is given by $a(p) = c(w'+d'p+p(p-1))$.}
\begin{lemma}\label{lem.general}
For any task that follows the general model, there exists a processor allocation that achieves \new{$\alpha^{\GEN}=2
$ and $\beta^{\GEN}=\frac{27}{13}$}.
\end{lemma}

\begin{proof}
If we allow the processor allocation to take non-integer values and assuming unbounded $\bar{p}$, the execution time function $t(p)$ would be minimized at $p^*=\sqrt{w'}$. Thus, the minimum execution time should satisfy $t^{\min}\ge c(2\sqrt{w'}+d'-1)$. Note that this bound will hold true regardless of the value of $\bar{p}$: it is obviously true if $\bar{p} \ge p^*$, otherwise $t^{\min}$ is achieved at $\bar{p}$, with a value also higher than $c(2\sqrt{w'}+d'-1)$.
Furthermore, the minimum area is obtained with one processor, i.e., $a^{\min} = a(1) = c(w'+d')$.

Recall that $p^{\max}$ denotes the number of processors that minimizes the execution time, i.e., $t(p^{\max}) = t^{\min}$. Clearly, we have either \neww{$p^{\max}=\bar{p}$} or $\lfloor \sqrt{w'} \rfloor \leq p^{\max} \leq \lceil \sqrt{w'} \rceil$.

We consider three cases.

\textbf{Case 1}: $w'\leq 4$ or \neww{$\bar{p}=1$}. In this case, it must be that $p^{\max}\leq 2$. We can then set $p = 1$, and get $\frac{a(p)}{a^{\min}}=1$ and $\frac{t(p)}{t^{\min}} \leq 2$.

\textbf{Case 2}: $4<w' \leq 49$ and \neww{$\bar{p}\geq 2$}. In this case, we set $p=2$ and get $\frac{a(p)}{a^{\min}} \leq \frac{w'+2d'+2}{w'+d'} \leq 2$. Similarly to the proof of Lemma \ref{lem.comm} (for the communication model), we can show that $f(w',p)\triangleq \frac{t(p)}{t^{\min}}=\frac{\frac{w'}{p}+d'+p-1}{\frac{w'}{p^{\max}}+d'+p^{\max}-1}$ is increasing with $w'$ if $w'\ge p^2$. Therefore, we can get $\frac{t(p)}{t^{\min}} \leq f(49,2) \leq \frac{\frac{49}{2}+d'+1}{2\sqrt{49}+d'-1} = \frac{51+2d'}{26+2d'} \leq 2$.

\textbf{Case 3}: $w'>49$ \neww{and $\bar{p}\geq 2$}. In this case, we will set $p=\min\left(\left\lfloor \frac{w'+d'}{\sqrt{w'}+d'} +\frac{1}{2}\right\rfloor,\bar{p}\right)$ and get:
\begin{align*}
\frac{a(p)}{a^{\min}} &= \frac{w'+p(d'+p-1)}{w'+d'} \\
&\leq \frac{w'+\left(\frac{w'+d'}{\sqrt{w'}+d'}+\frac{1}{2}\right)\left(d'+\frac{w'+d'}{\sqrt{w'}+d'}-\frac{1}{2}\right)}{w'+d'} \\
&= \frac{w'+\frac{d'}{2}-\frac{1}{4}+\frac{w'+d'}{\sqrt{w'}+d'}\left(d'+\frac{w'+d'}{\sqrt{w'}+d'} \right)}{w'+d'} \\
&\le \neww{\frac{w'+d'+\frac{w'+d'}{\sqrt{w'}+d'}\left(d'+\frac{w'+d'}{\sqrt{w'}+d'} \right)}{w'+d'}} \\
&= 1+\frac{d'(\sqrt{w'}+d')+w'+d'}{(\sqrt{w'}+d')^2} \\
&=1+\frac{d'^2+d'\sqrt{w'}+d'+w'}{d'^2+2d'\sqrt{w'}+w'} \\
&\leq 2 \ .
\end{align*}
The last inequality above comes from $w'>1$ and $d'>0$.

Since $w'>1$, we get $t^{\min}\ge c(2\sqrt{w'}+d'-1)>c(\sqrt{w'}+d')$. To derive the execution time ratio, we further consider two subcases.
\begin{itemize}
\item If $p=\left\lfloor \frac{w'+d'}{\sqrt{w'}+d'} +\frac{1}{2}\right\rfloor$, then $p\geq \frac{w'+d'}{\sqrt{w'}+d'} - \frac{1}{2} \geq \frac{w'-\frac{1}{2}\sqrt{w'}}{\sqrt{w'}+d'}$. We can then get:
\begin{align*}
\frac{t(p)}{t^{\min}} &\leq \frac{\frac{w'}{p}+d'+p-1}{\sqrt{w'}+d'} \\
&\leq \frac{\frac{w'(\sqrt{w'}+d')}{w'-\frac{1}{2}\sqrt{w'}}}{\sqrt{w'}+d'}+\frac{d'+\frac{w'+d'}{\sqrt{w'}+d'}}{\sqrt{w'}+d'} \\
&\leq \frac{1}{1-\frac{1}{2\sqrt{w'}}}+\frac{d'(\sqrt{w'}+d')+w'+d'}{(\sqrt{w'}+d')^2}\\
&\leq \frac{1}{1-\frac{1}{2\sqrt{w'}}}+1
\end{align*}
For the last inequality, we recognize the same term we had when bounding the area ratio, which is at most $1$. Finally, the last expression above \neww{decreases} with $w'$, so using $w'>49$, we get $\frac{t(p)}{t^{\min}} \le \frac{1}{1-\frac{1}{14}}+1=\frac{27}{13}$.
\item If $p = \bar{p} < \left\lfloor \frac{w'+d'}{\sqrt{w'}+d'} +\frac{1}{2}\right\rfloor$, \neww{and since $\bar{p}$ is an integer, then it is necessarily the case that $\bar{p} \le \left\lfloor \frac{w'+d'}{\sqrt{w'}+d'} +\frac{1}{2}\right\rfloor - 1 \le \frac{w'+d'}{\sqrt{w'}+d'} \leq \sqrt{w'}$ (because $w'>1$). Therefore, we should also have $p^{\max} = \bar{p} = p$, and thus $\frac{t(p)}{t^{\min}}=1$.} \qedhere
\end{itemize}
\end{proof}

\begin{theorem}\label{thm.general}
Algorithm \ref{alg.online} is \new{$\frac{27}{33-\sqrt{738}}<4.63$}-competitive for any graph of tasks that follow the general speedup model given in
Equation~\eqref{eq.exec_time}. This is achieved with $\mu^{\GEN}= \frac{33-\sqrt{738}}{27} \approx 0.216$.
\end{theorem}

\begin{proof}
\new{With $\alpha=2$ and $\beta=\frac{27}{13}$, we get $\mu=\frac{\alpha+\beta+1-\sqrt{(\alpha+\beta+1)^2-4\beta}}{2\beta}=\frac{33-\sqrt{738}}{27}$ and the constraint $\beta \ge \frac{\mu(\alpha-1)}{(1-\mu)^2} \approx 0.35$ is satisfied. Thus, the competitive ratio is given by $\frac{1}{\mu}= \frac{27}{33-\sqrt{738}} < 4.63$.}
\end{proof}

Finally, Table \ref{tab.all-values} summarizes the parameters and competitive ratios derived for all the considered speedup models.

\renewcommand{\arraystretch}{1.7}
\begin{table*}[h]
\centering
\caption{Summary of parameters and competitive ratios for different speedup models.}
\label{tab.all-values}
\begin{tabular}{| c | c | c | c | c | c |}
\hline
\textbf{Model $M$} & \textbf{$\mu^M$} & \textbf{$\alpha^M$} & \textbf{$\beta^M$} & \textbf{Comp. Ratio} \\
\hline
Roofline (\ROO) & $\frac{3-\sqrt{5}}{2} \approx 0.382$ & $1$ & $1$ & $\frac{2}{3-\sqrt{5}} \approx 2.62$\\
\hline
Comm. (\COM) & $\frac{23-\sqrt{313}}{18} \approx 0.295$ & $\frac{4}{3}$ & $\frac{3}{2}$ &$\frac{18}{23-\sqrt{313}} \approx 3.39$ \\
\hline
Amdahl (\AMD) & $\frac{1-\sqrt{8\sqrt{2}-11}}{2} \approx 0.22$  & $\frac{\sqrt{2}+1+\sqrt{2\sqrt{2}-1}}{2} \approx 1.88$ & $\frac{1+\sqrt{4\sqrt{2}+5}}{2} \approx 2.13$ & $\frac{2}{1-\sqrt{8\sqrt{2}-11}} \approx 4.55$ \\
\hline
General (\GEN) & $\frac{33-\sqrt{738}}{27} \approx 0.216$ & $2$ & $\frac{27}{13}$ & $\frac{27}{33-\sqrt{738}} \approx 4.63$\\
\hline
\end{tabular}
\end{table*}

\section{Lower Bounds for Online List Scheduling Algorithms with Deterministic Local Decisions}
\label{sec.low}

In Section~\ref{sec.comp-ratio}, we derived the competitive ratios of our online algorithm under several common speedup models. In this section, we will show corresponding lower bounds on the competitive ratios. We point out that, in contrast to the lower bounds proven in our preliminary work \cite{Benoit22_online}, which apply only to the presented algorithm, the lower bounds shown in this section are stronger, as they apply to \emph{any} online list scheduling algorithm with deterministic local decisions for processor allocation.

\begin{definition}
\label{localDet}
An online algorithm is said to make \textbf{deterministic local decisions} if it allocates processors by considering only the total number of processors (i.e., $P$) and the parameters of a task's speedup function (i.e., $w$, $\bar{p}$, $d$, $c$). Thus, two identical tasks will receive exactly the same allocation regardless of their relative positions in the task graph as well as the graph structure.
\end{definition}

Ultimately, we will show that our algorithm has the optimal competitive ratios over all algorithms in this class for the roofline, communication, and Amdahl's models. The result also indicates that our algorithms's competitive ratio for the general model is close to optimal using the lower bound of the Amdahl's model.

\subsection{Analysis Overview}
Under any model $M$, we have shown in Section \ref{sec.alg} that, for any task, our online algorithm achieves:
\begin{align}\label{eq.lower1}
\frac{a}{a^{\min}} \leq \alpha^M \text{~~and~~} \frac{t}{t^{\min}} \leq \beta^M \ ,
\end{align}
where $\frac{\alpha^M}{1-\mu^M}+\beta^M = \frac{1}{\mu^M}$. In particular, for any possible instance consisting of a set $\mathcal{T}$ of tasks, if our algorithm achieves a makespan of $T$ and the optimal makespan is $T^{\opt}$, then we have shown that:
\begin{align}\label{eq.lower2}
\frac{T}{T^{\opt}} \leq \max_{j \in \mathcal{T}} \frac{a_j}{a_j^{\min}(1-\mu^M)}+\max_{j \in \mathcal{T}}\frac{t_j}{t_j^{\min}} \le \frac{\alpha^M}{1-\mu^M}+\beta^M = \frac{1}{\mu^M} \ .
\end{align}

\subsubsection*{Two-step Approach}
To prove the lower bounds, we will proceed in two steps corresponding to proving the tightness of the above two inequalities (\ref{eq.lower1}) and (\ref{eq.lower2}), respectively. More precisely, we will fix a model $M$ and suppose an online list scheduling algorithm $\mathcal{A}$ respecting Definition~\ref{localDet} and having a competitive ratio strictly less than $\frac{1}{\mu^M}$ exists. Specifically, we will assume that $\mathcal{A}$'s competitive ratio is $\frac{1}{\mu^M}-4\epsilon$ for some $0< \epsilon < 1$. In the first step (Section~\ref{sec:local}), we will show the existence of two tasks $A$ and $B$ as well as another algorithm $\mathcal{A}^*$ such that:
\begin{align}
\frac{a_B}{a_B^*(1-\mu^M)}+\frac{t_A}{t_A^*} \geq \frac{1}{\mu^M}-\epsilon \ ,
\end{align}
thus showing the tightness of our local analysis. In the second step (Section~\ref{sec:global}), we will build an instance using these two tasks (as well as two other tasks $C$ and $D$) such that, by using list scheduling, algorithm $\mathcal{A}$ has no choice but to achieve a makespan that satisfies:
\begin{align}
\frac{T}{T^*} \geq \frac{a_B}{a_B^*(1-\mu^M)}+\frac{t_A}{t_A^*}-2\epsilon \geq \frac{1}{\mu^M}-3\epsilon \ ,
\end{align}
where $T$ and $T^*$ denote the makespans of $\mathcal{A}$ and $\mathcal{A}^*$ for this instance, respectively. This contradicts the assumed competitive ratio, thus showing the tightness of our global analysis and hence the non-existence of algorithm $\mathcal{A}$.

\subsubsection*{Notations}
We will use four different tasks $A, B, C, D$ to construct the lower bound instances. For algorithm $\mathcal{A}$, we let $p_A$ (resp. $p_B, p_C, p_D$) denote its processor allocation for task $A$ (resp. $B$, $C$, $D$), let $t_A$ (resp. $t_B, t_C, t_D$) denote the resulting execution time of the tasks, and let $a_A=t_A p_A$ (resp. $a_B, a_C, a_D$) denote the resulting area of the tasks. Similarly, we use $p^*_A, p^*_B, p^*_C, p^*_D, t^*_A, t^*_B, t^*_C, t^*_D, a^*_A, a^*_B, a^*_C, a^*_D$ to denote the corresponding values for algorithm $\mathcal{A^*}$.

\subsubsection*{Constraints}
The lower bound instances need to respect a set of constraints (or rules) for the tasks and for the task graph, which are required to show the global results. This will allow us to prove the lower bound regardless of the model, as long as these constraints are satisfied. For convenience, Table \ref{Constraints} lists and labels all the required constraints ($R$'s) along with some definitions ($F$'s). We will refer to them according to their labels, and use ($R:\checkmark$) to denote that a constraint $R$ is satisfied in the subsequent analysis.

\begin{table}[t]
\center
\caption{List of constraints ($R$'s) and definitions ($F$'s) for constructing lower bound instances.}
\label{Constraints}
\begin{tabular}{|p{1.2in}|l|}
\hline
\textbf{For tasks} & \textbf{For task graph} \\
\hline
 $p_A^* \leq P^{3/4}$ \hfill ($R_1$) & $P\geq \left(\frac{120900}{\epsilon}\right)^4$ \hfill ($F_{1})$ \\
\hline
   $0.1 \leq t_B^* \leq 100 $\hfill ($R_2$) & $X=\left\lceil\frac{P-p_C+1}{p_B} \right\rceil$ \quad \hfill ($F_{2}$) \\
\hline
  $p_B \leq P^{3/4}$ \hfill ($R_3$)  & $K=\left\lceil \frac{5t_A^*}{\epsilon X t_B^*} \right\rceil$ \quad \hfill ($F_{3}$) \\
\hline
 $t_D \leq t_B$ \hfill ($R_4$) & $Y=\left\lfloor \frac{XKt^*_B}{t^*_A} \right\rfloor$ \hfill ($F_{4}$) \\
\hline
 $t^*_D \leq \frac{\epsilon}{121P^2}$ \hfill ($R_{5}$)  & $Z=K(P-p^*_A)$ \quad  \hfill ($F_{5}$) \\
\hline
 $p_D \leq 4$ \hfill ($R_{6}$) &  \\
\hline
 $t^*_A \leq 24t^*_B$ \hfill ($R_7$) & $1 \leq X \leq P$ \hfill $(R_{13})$ \\
\hline
 $t^*_B=a^*_B=t_B(1)$ \hfill ($R_8$)  & $XKt^*_B \left(1-\frac{\epsilon}{5}\right) \leq Yt^*_A \leq XKt^*_B$ \quad \hfill ($R_{14}$) \\
\hline
  $t_A \leq 5t^*_A$ \hfill ($R_9$) & $K(P-P^{3/4}) \leq Z \leq \frac{121P}{\epsilon}$ \hfill  ($R_{15}$) \\
\hline
  $a_B \leq 5a_B^*$ \hfill ($R_{10}$)  &  \\
\hline
  $t^*_C \leq \frac{\epsilon}{121 P^2}$ \hfill ($R_{11}$) &  \\
\hline
  $p_C \geq \mu^M P$ \hfill ($R_{12}$) &  \\
\hline
\end{tabular}
\end{table}

\subsection{Step 1: Local Analysis}
\label{sec:local}
In this section, we will show that, for a given model $M$ and any online list scheduling algorithm $\mathcal{A}$ respecting Definition~\ref{localDet}, there exist tasks $A$ and $B$ as well as another algorithm $\mathcal{A^*}$ such that $\frac{a_B}{a^*_B (1-\mu^M)}+\frac{t_A}{t_A^*} \geq \frac{1}{\mu^M}-\epsilon$. We start with the following theorem and will prove it separately for each considered model.

\begin{theorem}
\label{th.mainLocal}
Given a model $M\in \{\ROO, \COM, \AMD\}$, let $\mathcal{A}$ be an online list scheduling algorithm respecting Definition~\ref{localDet} with a competitive ratio of $\frac{1}{\mu^M}-4\epsilon$ for some $0<\epsilon<1$, and let $P \geq \left(\frac{120900}{\epsilon}\right)^4$. Then, there exist four tasks $A$, $B$, $C$ and $D$ satisfying the constraints on tasks in Table~\ref{Constraints} ($R_1$ to $R_{12}$) and another algorithm $\mathcal{A}^*$ such that:
\begin{equation}
\label{eq.mainlocal}
\frac{a_B}{a_B^*(1-\mu^M)}+\frac{t_A}{t^*_A} \geq \frac{1}{\mu^M}-\epsilon \ .
\end{equation}
\end{theorem}

\begin{proof}
First, we set $t_C(p)=\frac{\epsilon}{121 P^2 \cdot p}$. Indeed, this execution time function belongs to all speedup models\footnote{For all models, we have $w = \frac{\epsilon}{121P^2}$. Additionally, for the roofline model, $\bar{p} = \infty$; for the communication model, $c=0$; and for the Amdhal's model, $d=0$.}, and clearly, we have ($R_{11}:\checkmark$), i.e., constraint $(R_{11})$ is satisfied. Further, if we had $p_C<\mu^M P$, then we would have $\frac{t_C}{t^*_C}>\frac{\frac{\epsilon}{121 P^2 \cdot \mu^{M}P}}{\frac{\epsilon}{121P^2\cdot P}}=\frac{1}{\mu^M}$, which contradicts the competitive ratio of $\mathcal{A}$ on an instance consisting of only one task $C$ ($R_{12}:\checkmark$).
Similarly, for task $A$, we must have $t_A \leq 5t_A^*$ to respect the competitive ratio of $\mathcal{A}$ on an instance consisting of a single such task, as $\frac{1}{\mu^M}<5$ for all models ($R_{9}:\checkmark$). For task $B$, we will set $p^*_B=1$ for all models ($R_8 \checkmark$).
Also, if we had $a_B>5a^*_B$, then an instance consisting of $P$ independent such tasks would result in a makespan at least $\frac{P a_B}{P}=a_B$ for $\mathcal{A}$, since $P a_B$ is the total area to be completed on $P$ processors, while $\mathcal{A^*}$ can execute all tasks simultaneously in parallel with a resulting makespan of $a_B^*$, which also contradicts the competitive ratio of $\mathcal{A}$ ($R_{10}:\checkmark$).

The following three lemmas will conclude the proof of the theorem by considering each of the three models separately. Note that we only need to define tasks $A$, $B$ and $D$,  and verify the constraints ($R_1$) to ($R_7$).
\end{proof}

\begin{lemma}
Theorem~\ref{th.mainLocal} is true for the roofline model.
\end{lemma}

\begin{proof}
For the roofline model, we will only use sequential tasks with $\bar{p}=1$, and set $t_A(p)=t_B(p)=1$ and $t_D(p)=\frac{\epsilon}{121P^2}$ for all $p$. Thus, we have ($R_2, R_4, R_5, R_7: \checkmark$). Clearly, if $\mathcal{A}$ allocates $3$ or more processors to task $B$ or $D$, then running $P$ independent such tasks would result in a makespan at least $3$ times that of the optimal using a single processor, constradicting the competitive ratio of $\mathcal{A}$. Thus, we can assume that $p_B \le 2 \leq P^{3/4}$ ($R_3: \checkmark$) and $p_D \le 2 \leq 4$ ($R_6: \checkmark$). We further set $p^*_A = p^*_B=1$ ($R_1: \checkmark$). These give $\frac{a_B}{a^*_B}\ge 1$ and $\frac{t_A}{t^*_A}= 1$. With $\mu^{\ROO}=\frac{3+\sqrt{5}}{2}$, we obtain:
\begin{equation*}
\frac{a_B}{a_B^*(1-\mu^{\ROO})}+\frac{t_A}{t^*_A} \geq \frac{1}{1-\mu^{\ROO}}+1 = \frac{1}{\mu^{\ROO}} \ .  \qedhere
\end{equation*}
\end{proof}

\begin{lemma}
Theorem~\ref{th.mainLocal} is true for the communication model.
\end{lemma}

\begin{proof}
Given algorithm $\mathcal{A}$, we define $\mathcal{U}$ to be the subset of $\{x \in \mathbb{R}^+\}$ such that $\mathcal{A}$ allocates one processor to a task whose execution time function has the form: $t(p)=\frac{x}{p}+p-1$. By definition, $\mathcal{A}$ always has the same allocation for identical tasks, so $\mathcal{U}$ is well-defined and must satisfy:
\begin{itemize}
\item $[0,0.1] \subseteq \mathcal{U}$, otherwise there would exist an $x \leq 0.1$ such that $\mathcal{A}$ allocates at least two processors for $t(p)=\frac{x}{p}+p-1$, and we would have $\frac{a}{a^{\min}} \geq \frac{a(2)}{a(1)}=\frac{x+2}{x}=1+\frac{2}{x} \geq 21$, which contradicts the competitive ratio of $\mathcal{A}$.
\item $\mathcal{U} \subseteq [0,64]$, otherwise there would exist an $x>64$ such that $\mathcal{A}$ allocates one processor for $t(p)=\frac{x}{p}+p-1$, and we would have $\frac{t}{t^{\min}}\geq \frac{x}{x/8+7}=\frac{8x+448}{x+56}-\frac{448}{x+56} \geq 8-\frac{448}{120}>4$, which also contradicts the competitive ratio of $\mathcal{A}$.
\end{itemize}

Based on the previous analysis, we now consider $s=sup(\mathcal{U}) \in [0.1,64]$. By definition, there exists a $\delta \in [0, \frac{1}{P})$ such that $(s-\delta) \in \mathcal{U}$ and $(s-\delta+\frac{1}{P}) \notin \mathcal{U}$. We choose such $\delta$, and set $\bar{w}=s-\delta > 0.09$, $t_A(p)=\frac{\bar{w}}{p}+p-1$ and $t_B(p)=\frac{\bar{w}+\frac{1}{P}}{p}+p-1$ with $p_A=1$ and $p_B>1$. Thus, we get $t^*_B = \bar{w}+\frac{1}{P}$ ($R_2: \checkmark$).

We also have $\frac{a_B}{a^{\min}_B}=\frac{\bar{w}+\frac{1}{P}+p_B(p_B-1)}{\bar{w}+\frac{1}{P}}\ge 1+\frac{(p_B-1)^2}{65}$, from which we can assume $p_B \leq 15$, otherwise $\frac{a_B}{a^{\min}_B}>4$, again contradicting $\mathcal{A}$'s competitive ratio. As $P > 81$, it leads to $p_B \leq P^{3/4}$ ($R_3: \checkmark$). We now show that Inequality~(\ref{eq.mainlocal}) is true:

\begin{align*}
\frac{a_B}{a_B^*(1-\mu^{\COM})}+\frac{t_A}{t^*_A} &\geq \frac{\bar{w}+2}{ (\bar{w}+\frac{1}{P})(1-\mu^{\COM})}+\frac{\bar{w}}{\frac{\bar{w}}{p^*_A}+p^*_A-1} \\
&= \frac{\bar{w}+2}{\bar{w}(1-\mu^{\COM})}\cdot \frac{1}{1+\frac{1}{P \bar{w}}}+\frac{\bar{w}}{\frac{\bar{w}}{p^*_A}+p^*_A-1} \\
&\geq \frac{\bar{w}+2}{\bar{w}(1-\mu^{\COM})}\cdot \left(1-\frac{1}{P \bar{w}}\right)+\frac{\bar{w}}{\frac{\bar{w}}{p^*_A}+p^*_A-1} \\
&= \frac{\bar{w}+2}{\bar{w}(1-\mu^{\COM})}-\frac{(\bar{w}+2)}{P \bar{w}^2(1-\mu^{\COM})}+\frac{\bar{w}}{\frac{\bar{w}}{p^*_A}+p^*_A-1} \\
&\geq \frac{\bar{w}+2}{\bar{w}(1-\mu^{\COM})}-\frac{2(\bar{w}+2)}{P \bar{w}^2}+\frac{\bar{w}}{\frac{\bar{w}}{p^*_A}+p^*_A-1} \\
&\geq \frac{\bar{w}+2}{\bar{w}(1-\mu^{\COM})}-\frac{16297}{P}+\frac{\bar{w}}{\frac{\bar{w}}{p^*_A}+p^*_A-1} \\
&\geq \frac{1+\frac{2}{\bar{w}}}{1-\mu^{\COM}}+\frac{1}{\frac{1}{p^*_A}+\frac{p^*_A-1}{\bar{w}}}-\epsilon \ .
\end{align*}

We further consider two cases:

\textbf{Case 1: $\bar{w} \leq 6$}. In this case, we set $p^*_A=2$ and obtain: 
\begin{align*}
\frac{a_B}{a_B^*(1-\mu^{\COM})}+\frac{t_A}{t^*_A} &\geq \frac{1+\frac{2}{\bar{w}}}{1-\mu^{\COM}}+\frac{1}{\frac{1}{2}+\frac{1}{\bar{w}}}-\epsilon \ .
\end{align*}
We define $f(\bar{w})\triangleq\frac{1+\frac{2}{\bar{w}}}{1-\mu^{\COM}}+\frac{1}{\frac{1}{2}+\frac{1}{\bar{w}}}$, and get $f'(\bar{w})=\frac{4}{(\bar{w}+2)^2}-\frac{2}{(1-\mu^{\COM})\bar{w}^2}$, which is negative in $[0,w^0)$, where $w^0$ is the smallest positive $\bar{w}$ such that $f'(\bar{w})=0$. Solving the equation above, we find $w^0=\frac{2+2\sqrt{2-2\mu^{\COM}}}{1-2\mu^{\COM}}>6$. Therefore, we can replace $\bar{w}$ by $6$ to obtain:
\begin{align*}
\frac{a_B}{a_B^*(1-\mu^{\COM})}+\frac{t_A}{t^*_A} &\geq \frac{4}{3(1-\mu^{\COM})}+\frac{3}{2}-\epsilon \\
&= \frac{\alpha^{\COM}}{1-\mu^{\COM}}+\beta^{\COM}-\epsilon \\
&= \frac{1}{\mu^{\COM}}-\epsilon \ .
\end{align*}

\textbf{Case 2: $\bar{w} > 6$}. In this case, we set $p^*_A=3$ and obtain: 
\begin{align*}
\frac{a_B}{a_B^*(1-\mu^{\COM})}+\frac{t_A}{t^*_A} &\geq \frac{1+\frac{2}{\bar{w}}}{1-\mu^{\COM}}+\frac{1}{\frac{1}{3}+\frac{2}{\bar{w}}}-\epsilon \ .
\end{align*}
We again define $f(\bar{w})\triangleq\frac{1+\frac{2}{\bar{w}}}{1-\mu^{\COM}}+\frac{1}{\frac{1}{3}+\frac{2}{\bar{w}}}$, and get $f'(\bar{w})=\frac{18}{(\bar{w}+6)^2}-\frac{2}{(1-\mu^{\COM})\bar{w}^2}$, which is positive in $(w^0,\infty)$, where $w^0$ is the largest positive $\bar{w}$ such that $f'(\bar{w})=0$. Solving the equation above, we find $w^0=\frac{6+18\sqrt{1-\mu^{\COM}}}{8-9\mu^{\COM}}<6$. Therefore, we can replace $\bar{w}$ by $6$ to obtain:
\begin{align*}
\frac{a_B}{a_B^*(1-\mu^{\COM})}+\frac{t_A}{t^*_A} &\geq \frac{4}{3(1-\mu^{\COM})}+\frac{3}{2}-\epsilon \\
&= \frac{\alpha^{\COM}}{1-\mu^{\COM}}+\beta^{\COM}-\epsilon \\
&= \frac{1}{\mu^{\COM}}-\epsilon \ .
\end{align*}

In both cases, we get the desired result with ($R_1: \checkmark$). Further, because $p^*_A \leq 3$ and $\bar{w} > 0.09$, we get $t^*_A \leq \bar{w}+2 \leq 24\bar{w} \leq 24t^*_B$ ($R_7: \checkmark$).

Finally, for task $D$, we set $t_D(p)=\frac{\epsilon}{121P^2 \cdot p}+p-1$ and $p_D^* =1$. Thus, $\mathcal{A}$ must also allocate one processor to the task (i.e., $P_D = 1$), otherwise $\frac{a_D}{a_{D}^{\min}}\ge \frac{\frac{\epsilon}{121P^2}+2}{\frac{\epsilon}{121P^2}}>5$, which contradicts its competitive ratio. Therefore, we have ($R_4, R_5, R_6: \checkmark$).
\end{proof}

\begin{lemma}
Theorem~\ref{th.mainLocal} is true for the Amdahl's model.
\end{lemma}

\begin{proof}
Given algorithm $\mathcal{A}$, we define $\mathcal{U}$ to be the subset of $\{x \in \mathbb{R}^+\}$ such that $\mathcal{A}$ allocates strictly less than $\sqrt{P}$ processors to a task whose execution time function has the form: $t(p)=\frac{x}{p}+\frac{1}{\sqrt{P}}$. By definition, $\mathcal{A}$ always has the same allocation for identical tasks, so $\mathcal{U}$ is well-defined and must satisfy:
\begin{itemize}
\item $[0,0.1] \subseteq \mathcal{U}$, otherwise there would exist an $x \leq 0.1$ such that $\mathcal{A}$ allocates at least $\sqrt{P}$ processors for $t(p)=\frac{x}{p}+\frac{1}{\sqrt{P}}$, and we would have $\frac{a}{a^{\min}} \geq \frac{a(\sqrt{P})}{a(1)}=\frac{x+1}{x+\frac{1}{\sqrt{P}}}=1+\frac{1-\frac{1}{\sqrt{P}}}{x+\frac{1}{\sqrt{P}}} \geq 1+\frac{0.99}{0.11}=10$, which contradicts the competitive ratio of $\mathcal{A}$.
\item $\mathcal{U} \subseteq [0,10]$, otherwise there would exist an $x>10$ such that $\mathcal{A}$ allocates less than $\sqrt{P}$ processors for $t(p)=\frac{x}{p}+\frac{1}{\sqrt{P}}$, and we would have $\frac{t}{t^{\min}}\geq \frac{t(\sqrt{P})}{t(P)} = \frac{\frac{x}{\sqrt{P}}+\frac{1}{\sqrt{P}}}{\frac{x}{P}+\frac{1}{\sqrt{P}}} > \frac{x}{\frac{x}{\sqrt{P}}+1} =\frac{1}{\frac{1}{\sqrt{P}}+\frac{1}{x}}>\frac{1}{0.11}>9$, which also contradicts the competitive ratio $\mathcal{A}$.
\end{itemize}

Based on the previous analysis, we now consider $s=sup(\mathcal{U}) \in [0.1,10]$.  By definition, there exists a $\delta \in [0, \frac{1}{\sqrt{P}})$ such that $(s-\delta) \in \mathcal{U}$ and $(s-\delta+\frac{1}{\sqrt{P}}) \notin \mathcal{U}$. We choose such $\delta$, and set $\bar{w}=s-\delta>0.09$, $t_A(p)=\frac{\bar{w}}{p}+\frac{1}{\sqrt{P}}$ and $t_B(p)=\frac{\bar{w}+\frac{1}{\sqrt{P}}}{p}+\frac{1}{\sqrt{P}}$ with $p_A < \sqrt{P}$ and $p_B \geq \sqrt{P}$. Thus, we get $t^*_A \le \bar{w}+\frac{1}{\sqrt{P}}$ and $t^*_B = \bar{w}+\frac{2}{\sqrt{P}}$ ($R_2, R_7: \checkmark$).

We can further assume $p_B\le P^{3/4}$, otherwise we would have $\frac{a_B}{a^{\min}_B} \geq \frac{\bar{w}+\frac{P^{3/4}}{\sqrt{P}}}{\bar{w}+\frac{2}{\sqrt{P}}} \geq \frac{0.09+P^{1/4}}{11}>5$, which contradicts the competitive ratio of $\mathcal{A}$ ($R_3: \checkmark$). Finally, we set $p^*_A=\lfloor P^{3/4} \rfloor$ ($R_1: \checkmark$) and get:
\begin{align*}
\frac{a_B}{a_B^*(1-\mu^{\AMD})}+\frac{t_A}{t^*_A} &\geq \frac{\bar{w}+1}{ (\bar{w}+\frac{2}{\sqrt{P}})(1-\mu^{\AMD})}+\frac{\frac{\bar{w}}{\sqrt{P}}+\frac{1}{\sqrt{P}}}{\frac{\bar{w}}{P^{3/4}-1}+\frac{1}{\sqrt{P}}} \\
&= \frac{\bar{w}+1}{\bar{w}(1-\mu^{\AMD})}\cdot \frac{1}{1+\frac{2}{\bar{w} \sqrt{P}}}+\frac{\bar{w}+1}{1+\frac{\bar{w}}{P^{1/4}-\frac{1}{\sqrt{P}}}} \\
&\geq \frac{\bar{w}+1}{\bar{w}(1-\mu^{\AMD})}\left(1-\frac{2}{\bar{w} \sqrt{P}}\right)+(\bar{w}+1)\left(1-\frac{\bar{w}}{P^{1/4}-\frac{1}{\sqrt{P}}}\right) \\
&\geq \frac{\bar{w}+1}{\bar{w}(1-\mu^{\AMD})}+\bar{w}+1-\frac{22}{0.09^2\times 0.5\sqrt{P}}-\frac{110}{P^{1/4}-\frac{1}{\sqrt{P}}} \\
&\geq \frac{\bar{w}+1}{\bar{w}(1-\mu^{\AMD})}+\bar{w}+1-\epsilon  \ .
\end{align*}

We now set $x=\frac{\bar{w}+1}{\bar{w}}>1$, so $\frac{x}{x-1}=\bar{w}+1$. We can finally conclude that:
\begin{align*}
\frac{a_B}{a_B^*(1-\mu^{\AMD})}+\frac{t_A}{t^*_A} &\geq \frac{x}{1-\mu^{\AMD}}+\frac{x}{x-1} -\epsilon \\
&\geq \min_{x'>1} \left( \frac{x'}{1-\mu^{\AMD}}+\frac{x'}{x'-1} \right) -\epsilon \\
&= \frac{\alpha^{\AMD}}{1-\mu^{\AMD}}+\beta^{\AMD}-\epsilon \\
&= \frac{2}{1-\sqrt{8\sqrt{2}-11}}-\epsilon \ .
\end{align*}

To show the third step above, we can take the derivative of $\frac{x'}{1-\mu^{\AMD}}+\frac{x'}{x'-1}$ and show that its minimum is achieved when $x'$ satisfies $(x'-1)^2=1-\mu^{\AMD}$. This is equivalent to $x'=\alpha^{\AMD}$, because $\left(\alpha^{\AMD}-1\right)^2=\frac{\left(\sqrt{2}-1+\sqrt{2\sqrt{2}-1}\right)^2}{4}=\frac{1+(\sqrt{2}-1)\sqrt{2\sqrt{2}-1}}{2}=\frac{1+\sqrt{8\sqrt{2}-11}}{2}=1-\mu^{\AMD}=(x'-1)^2$. As $x'>1$, the only solution is $x'=\alpha^{\AMD}$. From the proof of Lemma \ref{lem.amdahl}, we also get that $\beta^{\AMD} = \frac{\alpha^{\AMD}}{\alpha^{\AMD}-1} = \frac{x'}{x'-1}$.

Finally, for task $D$, we set $t_D(p)=\frac{\epsilon}{121P^2}$ and $p^*_D=1$. Thus, we must have $p_D \leq 4$, otherwise $\frac{a_D}{a^{\min}_D}>4$, which contradicts the competitive ratio of $\mathcal{A}$ ($R_4, R_5, R_6: \checkmark$).
\end{proof}

\subsection{Step 2: Global Analysis}
\label{sec:global}

In this section, we assume that algorithm $\mathcal{A}$ and model $M$ are fixed, while the tasks $A$, $B$, $C$, $D$ and algorithm $\mathcal{A}^*$ are chosen such that the conditions of Theorem~\ref{th.mainLocal} hold. We construct a task graph (as shown in Figure~\ref{fig.instance}), based on which we will show that $\frac{T}{T^*} \geq  \frac{a_B}{a^*_B (1-\mu^M)}+\frac{t_A}{t_A^*} -2\epsilon \geq \frac{1}{\mu^M}-3\epsilon$.

In our constructed task graph, the tasks are partitioned into four different groups: $\mathcal{T}_A, \mathcal{T}_B$, $\mathcal{T}_C$ and $\mathcal{T}_D$. Specifically,
\begin{itemize}
  \item $\mathcal{T}_A$ has $Y$ tasks identical to $A$, labeled as $(A_i)_{i \in [1,Y]}$;
  \item $\mathcal{T}_B$ has $XZ$ tasks identical to $B$, labeled as $(B_{i,j})_{i \in [1,Z], j \in [1,X]}$;
  \item $\mathcal{T}_C$ has $Z$ tasks identical to $C$, labeled as $(C_i)_{i \in [1,Z]}$;
  \item $\mathcal{T}_D$ has $Z$ tasks identical to $D$, labeled as $(D_i)_{i \in [1,Z]}$,
\end{itemize}
where $X=\left\lceil\frac{P-p_C+1}{p_B} \right\rceil$, and using $K=\left\lceil \frac{5t_A^*}{\epsilon X t_B^*} \right\rceil$, we set $Y=\left\lfloor \frac{XKt^*_B}{t^*_A} \right\rfloor$ and $Z=K(P-p^*_A)$. These parameters are specified as definitions ($F$'s) in Table \ref{Constraints}.

The tasks are organized in layers and have the following precedence constraints:
\begin{itemize}
  \item task $C_i$ is the predecessor of task $D_{i+1}$ for $1 \leq i <Z$, and of tasks $B_{i+1,j}$ for $1 \leq i <Z$ and $1\leq j\leq X$;
  \item task $D_i$ is the predecessor of task $C_i$ for $1\leq i\leq Z$;
  \item task $C_Z$ is the predecessor of task $A_1$;
  \item task $A_i$ is the predecessor of task $A_{i+1}$ for $1\leq i < Y$.
\end{itemize}

\begin{figure}
\centering
\scalebox{0.8}{%
\begin{tikzpicture}[xscale=1,yscale=.8]
\tikzset{vertex/.style = {shape=circle,fill=blue!20,minimum size=.1cm}}
\tikzset{edge/.style = {->,> = latex'}}
\SetUpEdge[lw         = .5pt,
           color      = black,
           labelcolor = white]
\node[vertex,fill=yellow!40] (aa) at  (0,6) {$D_1$};
\node[vertex] (a) at  (1,6) {$C_1$};
\node[vertex,fill=yellow!40] (bb) at  (2,6) {$D_2$};
\node[vertex] (b) at  (3,6) {$C_2$};
\node[vertex,fill=white] (cc) at  (4,6) {$\cdots$} ;
\node[vertex,fill=yellow!40] (dd) at  (5,6) {$D_Z$};
\node[vertex] (d) at  (6,6) {$C_Z$};

\node[vertex,fill=red!20] (d1) at (8,6) {$A_1$};
\node[vertex,fill=red!20] (d2) at (8,4.6) {$A_2$};
\node[vertex,fill=white] (d3) at (8,3.3) {$\cdots$};
\node[vertex,fill=red!20] (d4) at (8,2) {$A_Y$};

\node[vertex,fill=green!20] (b11) at  (0,5) {\scriptsize{$B_{1,1}$}};
\node[vertex,fill=green!20] (b12) at  (0,4) {\scriptsize{$B_{1,2}$}};
\draw  (0,3) rectangle node {$\cdots$} (0,3) ;
\node[vertex,fill=green!20] (b1x) at  (0,2) {\scriptsize{$B_{1,X}$}};

\node[vertex,fill=green!20] (b21) at  (2,5) {\scriptsize{$B_{2,1}$}};
\node[vertex,fill=green!20] (b22) at  (2,4) {\scriptsize{$B_{2,2}$}};
\node[vertex,fill=white] (b2) at  (2,3) {$\cdots$} ;
\node[vertex,fill=green!20] (b2x) at  (2,2) {\scriptsize{$B_{2,X}$}};

\node[vertex,fill=white] (bb1) at  (4,5) {$\cdots$} ;
\node[vertex,fill=white] (bb2) at  (4,4) {$\cdots$} ;
\node[vertex,fill=white] (bb3) at  (4,3) {$\cdots$} ;
\node[vertex,fill=white] (bb4) at  (4,2) {$\cdots$} ;

\node[vertex,fill=green!20] (by1) at  (6,5) {\scriptsize{$B_{Z,1}$}};
\node[vertex,fill=green!20] (by2) at  (6,4) {\scriptsize{$B_{Z,2}$}};
\node[vertex,fill=white] (b2) at  (6,3) {$\cdots$} ;
\node[vertex,fill=green!20] (byx) at  (6,2) {\scriptsize{$B_{Z,X}$}};

\tikzset{EdgeStyle/.style={->,> = latex'}}
\Edge(a)(bb)
\Edge(aa)(a)
\Edge(bb)(b)
\Edge(b)(cc)
\Edge(cc)(dd)
\Edge(dd)(d)
\Edge(d)(d1)
\Edge(d1)(d2)
\Edge(d2)(d3)
\Edge(d3)(d4)
\Edge(a)(b21)
\Edge(a)(b22)
\Edge(a)(b2x)

\Edge(b)(bb1)
\Edge(b)(bb2)
\Edge(b)(bb4)
\Edge(dd)(by1)
\Edge(dd)(by2)
\Edge(dd)(byx)
\end{tikzpicture}}
\caption{A task graph for proving lower bounds.}
\label{fig.instance}
\end{figure}
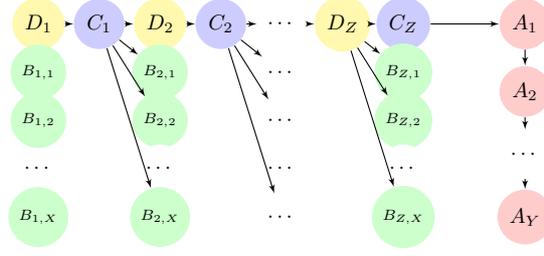

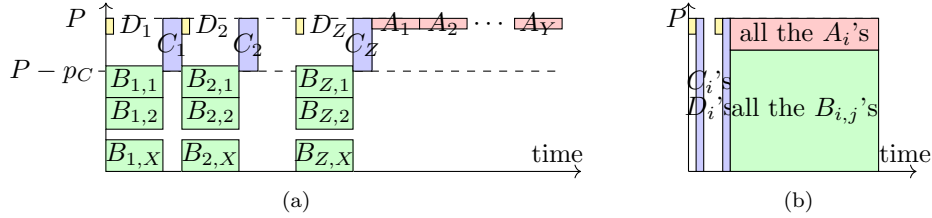
\begin{figure*}
\centering
\subfloat[][]{
\begin{tikzpicture}[xscale=0.5,yscale=.7]
\draw[->] (0,0) -- (12.5,0) node[pos=.95,above] {time} ;
\draw[->] (0,0) -- (0,3.2) ;
\draw[-,dashed] (0,2.9) -- (8,2.9)  node[pos=-.1]{$P$} ;
\draw[-,dashed] (-.2,1.9) -- (12,1.9) node[pos=-.1]{\new{$P-p_C$}};

\draw[fill=green!20]  (0, 0) rectangle node {$B_{1,X}$} (1.5, 0.6) ;
\draw[fill=green!20]  (0, 0.8) rectangle node {$B_{1,2}$} (1.5, 1.4) ;
\draw[fill=green!20]  (0, 1.4) rectangle node {$B_{1,1}$} (1.5, 2) ;
\draw[fill=yellow!40, anchor=west]  (0,2.6) rectangle node {$D_{1}$} (0.2, 2.9) ;
\draw[fill=blue!20] (1.5,1.9) rectangle node{$C_1$} (2,2.9) ;

\draw[fill=green!20]  (2, 0) rectangle node {$B_{2,X}$} (3.5, 0.6) ;
\draw[fill=green!20]  (2, 0.8) rectangle node {$B_{2,2}$} (3.5, 1.4) ;
\draw[fill=green!20]  (2, 1.4) rectangle node {$B_{2,1}$} (3.5, 2) ;
\draw[fill=yellow!40, anchor=west]  (2,2.6) rectangle node {$D_{2}$} (2.2, 2.9) ;
\draw[fill=blue!20] (3.5,1.9) rectangle node{$C_2$} (4,2.9) ;

\draw[fill=green!20]  (5, 0) rectangle node {$B_{Z,X}$} (6.5, 0.6) ;
\draw[fill=green!20]  (5, 0.8) rectangle node {$B_{Z,2}$} (6.5, 1.4) ;
\draw[fill=green!20]  (5, 1.4) rectangle node {$B_{Z,1}$} (6.5, 2) ;
\draw[fill=yellow!40, anchor=west]  (5,2.6) rectangle node {$D_{Z}$} (5.2, 2.9) ;
\draw[fill=blue!20] (6.5,1.9) rectangle node{$C_Z$} (7,2.9) ;

\draw[fill=red!20] (7,2.7) rectangle node{$A_1$} (8.25,2.9) ;
\draw[fill=red!20] (8.25,2.7) rectangle node{$A_2$} (9.5,2.9) ;
\draw[draw=none] (9.5,2.7) rectangle node{$\cdots$} (10.75,2.9) ;
\draw[fill=red!20] (10.75,2.7) rectangle node{$A_Y$} (12,2.9) ;
\end{tikzpicture}
}
\hspace{.5cm}
\subfloat[][]{
\begin{tikzpicture}[xscale=0.5,yscale=.7]
\draw[->] (0,0) -- (6,0) node[pos=.95,above] {time} ;
\draw[->] (0,0) -- (0,3.2) ;
\draw[-,dashed] (1,2.9) -- (5,2.9)  node[pos=-0.33]{$P$} ;

\draw[fill=yellow!40] (0,2.6) rectangle  (0.2,2.9) ;
\draw[fill=blue!20] (0.2,0) rectangle  (0.4,2.9) ;
\draw[fill=yellow!40] (0.7,2.6) rectangle  (0.9,2.9) ;
\draw[fill=blue!20] (0.9,0) rectangle  (1.1,2.9) ;
\draw[draw=none] (0,1) rectangle node{$D_i$'s} (1.1,1.4);
\draw[draw=none] (0,1.5) rectangle node{$C_i$'s} (1.1,1.9);

\draw[fill=green!20]  (1.1,0) rectangle node {all the $B_{i,j}$'s} (5, 2.3) ;

\draw[fill=red!20] (1.1,2.3) rectangle node{all the $A_i$'s} (5,2.9) ;
\end{tikzpicture}
}
\caption{Shapes of algorithm $\mathcal{A}$'s schedule (a) and algorithm $\mathcal{A^*}$'s schedule (b) for the task graph of Figure~\ref{fig.instance}.}
\label{fig.sched}
\end{figure*}

To prove the lower bound, we will first show that the constraints $(R_{13})$ to $(R_{15})$ in Table~\ref{Constraints} pertaining to the parameters of the task graph are also respected (Lemma \ref{lem.last_constraints}). We will then show that algorithm $\mathcal{A}$'s schedule must follow the shape as shown in Figure~\ref{fig.sched}(a) (Lemma~\ref{lem.doomed}), whereas algorithm $\mathcal{A^*}$ could wait until tasks in $\mathcal{T}_C$ and $\mathcal{T}_D$ are finished before launching tasks in $\mathcal{T}_A$ and $\mathcal{T}_B$, resulting in a better schedule as shown in Figure~\ref{fig.sched}(b). This last result together with Theorem~\ref{th.mainLocal} will lead to a contradiction, hence proving the lower bound (Theorem \ref{lem.LB}). In the following analysis, we will provide a reference to a constraint or a definition whenever it is used.

\begin{lemma}\label{lem.last_constraints}
Given the setting above, constraints $(R_{13})$, $(R_{14})$ and $(R_{15})$ in Table~\ref{Constraints} are satisfied.
\end{lemma}

\begin{proof}
Constraint $(R_{13})$ can be obtained directly from the definition of $X$:
\begin{align*}
1 \le X = \left\lceil \frac{P-p_C+1}{p_B} \right\rceil \le P \ . \quad\quad\quad (F_2)
\end{align*}

Constraint $(R_{14})$ can be derived from the definitions of $Y$ and $K$:
\begin{align*}
XK t^*_B \geq Yt^*_A &\geq XKt^*_B-t^*_A \quad\quad\quad\quad (F_4)\\
&= XKt^*_B\left(1-\frac{t^*_A}{XKt^*_B}\right) \\
&\geq XKt^*_B\left(1-\frac{\epsilon}{5}\right) \ .  \quad\quad (F_{3})
\end{align*}

Finally, constraint $(R_{15})$ can be obtained with:
\begin{align*}
K(P-P^{3/4}) \leq  Z &\leq KP~~~\quad\quad\quad\quad\quad\quad (F_5, R_1)\\
&\leq \left(\frac{5 t^*_A}{\epsilon X t^*_B}+1\right)P~~~  \quad \quad  (F_3) \\
&\leq \left(\frac{120}{\epsilon}+1\right)P  \quad\quad\quad  (R_7, R_{13}) \\
&\leq \frac{121P}{ \epsilon}  \ . \quad\quad\quad\quad\quad (\epsilon<1) \qedhere
\end{align*}
\end{proof}

\begin{lemma}\label{lem.doomed}
For a given task $\mathcal{T}$, let $s(\mathcal{T})$ denote its starting time in algorithm $\mathcal{A}$'s schedule and $e(\mathcal{T})$ its ending time. If all constraints in Table~\ref{Constraints} are satisfied, then algorithm $\mathcal{A}$'s schedule must follow the shape as shown in Figure~\ref{fig.sched}(a), i.e.:
\begin{itemize}
\item $s(D_i)=s(B_{i,1})=\cdots = s(B_{i,X}), \forall i \in [1,Z]$;
\item $s(C_i)=e(B_{i,1})=\cdots=e(B_{i,X}), \forall i \in [1,Z]$;
\item $s(D_i)=e(C_{i-1}), \forall i \in [2,Z]$;
\item $s(A_1)=e(C_Z)$;
\item $s(A_i)=e(A_{i-1}), \forall i \in [2,Y]$.
\end{itemize}
As a result, the makespan of $\mathcal{A}$ must satisfy: $T \geq Zt_B+Yt_A$.
\end{lemma}

\begin{proof}
It is possible to simultaneously run all the tasks $B_{i,j}$'s and $D_i$ in layer $i$, as the total number of processors required is:
\begin{align*}
Xp_B+p_D &\leq \left(\frac{P-p_C+1}{p_B}+1\right)P_B+4 \quad\quad (F_2, R_6) \\
&= P-P_C + P_B + 5 \\
&\leq (1-\mu^M)P +P^{3/4}+5~~  \quad\quad\quad (R_3, R_{12}) \\
&< 0.8P+0.01P+5 <P \ . \quad\quad\quad (F_1, \mu_M > 0.2)
\end{align*}

However, it is not possible to run all the $B_{i,j}$'s and $C_i$ in parallel, as the number of processors required would be:
\begin{align*}
Xp_B+p_C &\geq \frac{P-p_C+1}{p_B}p_B + p_C \quad\quad (F_2) \\
&= P+1 \ .
\end{align*}

Therefore, given that algorithm $\mathcal{A}$ uses list scheduling to schedule the tasks, we get $s(D_1)=s(B_{1,1})=\cdots = s(B_{1,X})$. Furthermore, since $t_D \leq t_B$ ($R_4$), $C_1$ becomes available before the first layer of tasks in $\mathcal{T}_B$ finishes and will be launched as soon as the layer is done, which gives $s(C_1)=e(B_{1,1})=\cdots = e(B_{1,X})$. A direct induction shows that, using list scheduling, the same scenario would happen for all the $Z$ layers. Finally, the tasks in $\mathcal{T}_A$ are executed one after another after the completion of $C_Z$, so the schedule corresponds to the exactly one shown in Figure~\ref{fig.sched}(a).
Since there are $Z$ layers of tasks in $\mathcal{T}_B$ and $Y$ layers of tasks in $\mathcal{T}_A$, the makespan of algorithm $\mathcal{A}$ must satisfy $T \geq Zt_B+Yt_A$.
\end{proof}

\begin{theorem}
\label{lem.LB}
Given a model $M \in \{\ROO, \COM, \AMD\}$, there is no online list scheduling algorithm respecting Definition~\ref{localDet} with a competitive ratio strictly less than $\frac{1}{\mu^M}$.
\end{theorem}

\begin{proof}
We prove the theorem by contradiction. Specifically, we assume that there exists such an algorithm $\mathcal{A}$ with a competitive ratio of $\frac{1}{\mu^M}-4\epsilon$ for some $0< \epsilon < 1$, as also assumed in Theorem \ref{th.mainLocal}. We will then show, using the constructed task graph, that the makespan of $\mathcal{A}$ satisfies $\frac{T}{T^*} \geq \frac{1}{\mu^M}-3\epsilon$, which leads to a contradiction, hence suggesting that no such algorithm should exist.

We first bound the makespan $T^*$ of algorithm $\mathcal{A}^*$, assuming that it follows the schedule of Figure~\ref{fig.sched}(b) by first running all the $C_i$'s and $D_i$'s before running the $A_i$'s sequentially while at the same time using one processor to run each of the $B_{i,j}$'s. As there are $XZ$ tasks of $B_{i,j}$'s, executing them all on $P-p^*_A$ processors takes time $\left\lceil \frac{XZ}{P-p^*_A} \right\rceil t_B^* = \left\lceil \frac{XK(P-p^*_A)}{P-p^*_A} \right\rceil t_B^* = XKt_B^*$~$(F_5)$, whereas executing all the $A_i$'s takes time $Yt^*_A \leq XKt^*_B$ ($R_{14})$. Therefore, $T^*$ should satisfy:
\begin{align*}
T^* &\leq Z(t^*_C+t^*_D)+XKt^*_B \\
&\leq \frac{121P}{\epsilon}\left(\frac{\epsilon}{121P^2}+\frac{\epsilon}{121P^2} \right)+XKt^*_B \quad\quad (R_5,R_{11},R_{15}) \\
&\leq \frac{2}{P}+XKt^*_B  \ .
\end{align*}

Now, using the result of Lemma~\ref{lem.doomed}, we get:
\begin{align*}
\frac{T}{T^*} &\geq \frac{Zt_B+Yt_A}{\frac{2}{P}+XKt_B^*} \\
&\geq \frac{K(P-P^{3/4})t_B}{2+ X K t^*_B}+\frac{Y t_A}{\frac{2}{P}+Y t^*_A/\left(1-\frac{\epsilon}{5}\right)}  \quad\quad\quad\quad\quad\quad\quad (R_{14},R_{15}) \\
&= \frac{(P-P^{3/4})t_B}{\frac{2}{K}+ X t^*_B}+\frac{t_A}{t^*_A/\left(1-\frac{\epsilon}{5}\right)+\frac{2}{YP}}  \\
&\geq \frac{(P-P^{3/4})t_B}{2+t^*_B\left(\frac{P(1-\mu^M)+1}{p_B}+1\right)}+\frac{t_A}{t^*_A}\cdot \frac{1-\frac{\epsilon}{5}}{1+\frac{2}{PXKt^*_B}(1-\frac{\epsilon}{5})} \quad (F_2, R_{12}, F_4) \\
&\geq \frac{(P-P^{3/4})t_B}{2+2t^*_B+\frac{ P t_B^*(1-\mu^M)}{p_B}}+\frac{t_A}{t^*_A}\cdot\frac{1-\frac{\epsilon}{5}}{1+\frac{2}{Pt^*_B}}  \\
&\geq \frac{t_B p_B}{\frac{p_B(2+2t^*_B)}{P-P^{3/4}}+\frac{P t^*_B(1-\mu^M)}{P-P^{3/4}}}+\frac{t_A}{t^*_A}\cdot\frac{1-\frac{\epsilon}{5}}{1+\frac{20}{P}}  \quad\quad\quad\quad\quad\quad (R_{2})\\
&\geq \frac{a_B}{\frac{2+2t^*_B}{P^{1/4}-1}+\frac{t^*_B(1-\mu^M)}{1-P^{-1/4}}}+\frac{t_A}{t^*_A}\cdot\frac{1-\frac{\epsilon}{5}}{1+\frac{20}{P}}  \quad\quad\quad\quad\quad\quad\quad\quad (R_{3})\\
&\geq \frac{a_B}{(4+4t^*_B)P^{-1/4}+t^*_B(1-\mu^M)(1+2P^{-1/4})}+\frac{t_A}{t^*_A}\cdot\frac{1-\frac{\epsilon}{5}}{1+\frac{20}{P}} \ .
\end{align*}

The last step above assumes $\frac{1}{P^{1/4}-1}\leq 2P^{-1/4}$ and $\frac{1}{1-P^{-1/4}}\leq 1+2P^{-1/4}$, both of which are true if $P^{1/4}\ge 2$, i.e., $P\geq 16$. We conclude with the following derivations:
\begin{align*}
\frac{T}{T^*} &\geq \frac{a_B}{t^*_B(1-\mu^M)+ (4+6t^*_B) P^{-1/4}}+\frac{t_A}{t^*_A}\cdot\frac{1-\frac{\epsilon}{5}}{1+\frac{20}{P}}   \\
&\geq \frac{a_B}{a^*_B(1-\mu^M)+ 604P^{-1/4}}+\frac{t_A}{t^*_A}\cdot\frac{1-\frac{\epsilon}{5}}{1+\frac{20}{P}}   \quad\quad\quad\quad\quad (R_2,R_8) \\
&= \frac{a_B}{a^*_B(1-\mu^M)}\cdot \frac{1}{1+\frac{604 P^{-1/4}}{a^*_B(1-\mu^M)}}+\frac{t_A}{t^*_A}\cdot \frac{1-\frac{\epsilon}{5}}{1+\frac{20}{P}}  \\
&\geq \frac{a_B}{a^*_B(1-\mu^M)}\cdot \frac{1}{1+12080P^{-1/4}}+\frac{t_A}{t^*_A}\cdot \frac{1-\frac{\epsilon}{5}}{1+\frac{20}{P}}   \quad\quad\quad (R_2, \mu^M \leq 0.5) \\
&\geq \frac{a_B}{a^*_B(1-\mu^M)}\left(1-12080P^{-1/4} \right)+\frac{t_A}{t^*_A}\left(1-\frac{20}{P} \right)\left(1-\frac{\epsilon}{5}\right)   \\
&\geq \frac{a_B}{a^*_B(1-\mu^M)}+\frac{t_A}{t^*_A}-120800P^{-1/4}-\frac{100}{P}-\epsilon  \quad\quad\quad (R_9, R_{10}, \mu^M \leq 0.5) \\
&\geq \frac{a_B}{a^*_B(1-\mu^M)}+\frac{t_A}{t^*_A}-120800P^{-1/4}-100P^{-1/4}-\epsilon   \\
&= \frac{a_B}{a^*_B(1-\mu^M)}+\frac{t_A}{t^*_A}-120900P^{-1/4}-\epsilon  \\
&\geq \frac{a_B}{a^*_B(1-\mu^M)}+\frac{t_A}{t^*_A}-2\epsilon  \quad\quad\quad\quad\quad\quad\quad\quad\quad\quad\quad\quad (F_1) \\
&\geq \frac{1}{\mu^M}-3\epsilon \nonumber \ .
\end{align*}

The last step above applies Theorem~\ref{th.mainLocal}, and the result proves this theorem and the optimal competitive ratio of our algorithm for these speedup models.
\end{proof}

\emph{Remarks.} Since the Amdahl's model is a special case of the general model, its lower bound also applies to the general model.

\section{A Lower Bound of Any Deterministic Online Algorithm for Arbitrary Speedup Model}
\label{sec.lower}

So far, we have focused on the general speedup model of Equation~\eqref{eq.exec_time} and its three special cases. In this section, we show that the competitive ratio of any deterministic online algorithm (including ours) can be unbounded under an \emph{arbitrary} speedup model\footnote{Under an arbitrary speedup model, the execution time $t(p)$ of a task can take any arbitrary function of its processor allocation $p$.}.

\begin{theorem}\label{arbitrary_lowerbound}
Any deterministic online algorithm is at least $\Omega(\ln(D))$-competitive for scheduling moldable task graphs under an arbitrary speedup model, where $D$ denotes the number of tasks along the longest (critical) path of the graph.
\end{theorem}

\begin{proof}
We fix an arbitrary integer $\ell>1$ and set $K=2^\ell$. The instance consists of $n=2^K-1$ independent linear task chains organized in groups. Specifically, for any $i \in [1,K]$, group $i$ contains $2^{K-i}$ linear chains, each with exactly $i$ tasks. Thus, the number of tasks along the longest path of the graph is given by $D = K$.  Figure~\ref{fig:work3} shows such an instance for $\ell=2, K = 4$ and $n = 15$.
All tasks in the graph are identical, with an execution time function $t(p)=\frac{1}{\lg(p)+1}$. Here, $\lg$ denotes logarithm to the base 2.
We set the total number of processors to be $P=K\cdot 2^{K-1}$.

\begin{figure}[t]
\begin{center}
\scalebox{0.7}{%
\begin{tikzpicture}[xscale=0.4,yscale=0.4]
\small
\draw[fill=red!20] (3,3) circle (30pt) node(A){15(1)};
\draw[fill=red!20] (8,3) circle (30pt) node(B){15(2)};
\draw[fill=red!20] (13,3) circle (30pt) node(C){15(3)};
\draw[fill=red!20] (18,3) circle (30pt) node(D){15(4)};

\draw[fill=green!20] (3,6) circle (30pt) node(E){14(1)};
\draw[fill=green!20] (8,6) circle (30pt) node(F){14(2)};
\draw[fill=green!20] (13,6) circle (30pt) node(G){14(3)};

\draw[fill=green!20] (3,9) circle (30pt) node(H){13(1)};
\draw[fill=green!20] (8,9) circle (30pt) node(I){13(2)};
\draw[fill=green!20] (13,9) circle (30pt) node(J){13(3)};

\draw[fill=yellow!20] (3,12) circle (30pt) node(K){11(1)};
\draw[fill=yellow!20] (8,12) circle (30pt) node(L){11(2)};
\draw[fill=yellow!20] (13,12) circle (30pt) node(M){12(1)};
\draw[fill=yellow!20] (18,12) circle (30pt) node(N){12(2)};

\draw[fill=yellow!20] (3,15) circle (30pt) node(O){9(1)};
\draw[fill=yellow!20] (8,15) circle (30pt) node(P){9(2)};
\draw[fill=yellow!20] (13,15) circle (30pt) node(Q){10(1)};
\draw[fill=yellow!20] (18,15) circle (30pt) node(R){10(2)};

\draw[fill=blue!20] (3,18) circle (30pt) node{5};
\draw[fill=blue!20] (8,18) circle (30pt) node{6};
\draw[fill=blue!20] (13,18) circle (30pt) node{7};
\draw[fill=blue!20] (18,18) circle (30pt) node{8};

\draw[fill=blue!20] (3,21) circle (30pt) node{1};
\draw[fill=blue!20] (8,21) circle (30pt) node{2};
\draw[fill=blue!20] (13,21) circle (30pt) node{3};
\draw[fill=blue!20] (18,21) circle (30pt) node{4};

\draw [decorate, decoration = {brace, amplitude=6pt}] (0,17) --  (0,22) node [black,midway,xshift=-1cm] {\footnotesize Group $1$};
\draw [decorate, decoration = {brace, amplitude=6pt}] (0,11) --  (0,16) node [black,midway,xshift=-1cm] {\footnotesize Group $2$};
\draw [decorate, decoration = {brace, amplitude=6pt}] (0,5) --  (0,10) node [black,midway,xshift=-1cm] {\footnotesize Group $3$};
\draw [decorate, decoration = {brace, amplitude=4pt}] (0,2) --  (0,4) node [black,midway,xshift=-1cm] {\footnotesize Group $4$};

\draw [->] (4.2,3) -- (6.8,3);
\draw [->] (4.2,6) -- (6.8,6);
\draw [->] (4.2,9) -- (6.8,9);
\draw [->] (4.2,12) -- (6.8,12);
\draw [->] (4.2,15) -- (6.8,15);
\draw [->] (9.2,3) -- (11.8,3);
\draw [->] (9.2,6) -- (11.8,6);
\draw [->] (9.2,9) -- (11.8,9);
\draw [->] (14.2,3) -- (16.8,3);
\draw [->] (14.2,12) -- (16.8,12);
\draw [->] (14.2,15) -- (16.8,15);

\end{tikzpicture}
}
\end{center}
\caption{A lower bound instance in Theorem \ref{arbitrary_lowerbound} with $\ell=2$, $K = 4$, and $n = 15$ linear task chains. Each circle represents a task and the number inside each circle indicates the ID of the linear chain the task is in (and the number in the parenthesis indicates the task's position in that linear chain). }\label{fig:work3}
\end{figure}
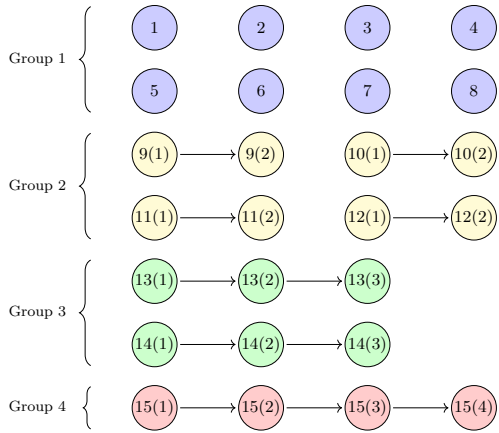

We show that the optimal offline algorithm completes the above instance with a makespan at most 1, whereas any deterministic online algorithm may produce a makespan at least $\ln(K)-\ln(\ell)-\frac{1}{\ell}$, thus proving the result.

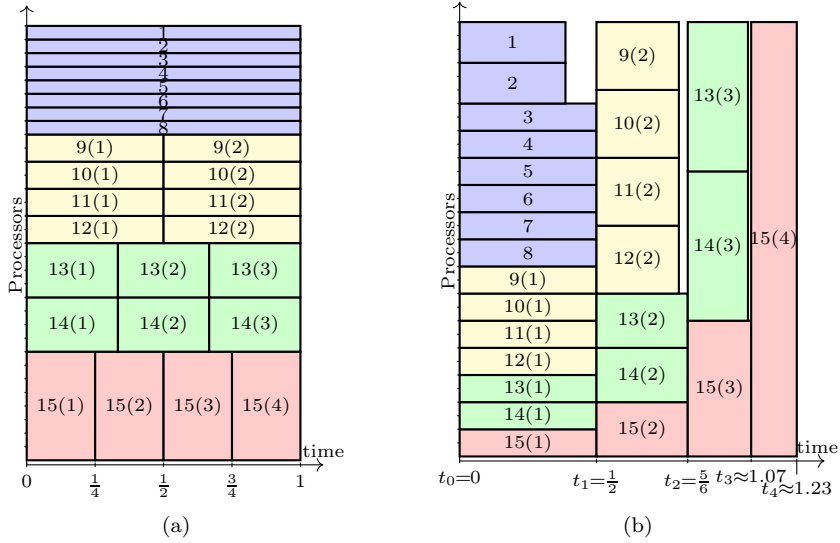
\begin{figure*}[t]
\centering
\subfloat[][]{
\begin{tikzpicture}[xscale=.3,yscale=0.3]
\scriptsize
\draw[->] (0,0) -- (13,0) node[pos=1,above] {time} ;
\draw[->] (0,0) -- (0,20.3) ; 
\draw (-.5,9.7) node[rotate=90] {Processors} ;
\draw[-]  (-.1,0.8) -- (.1,0.8) ;
\draw[-]  (-.1,1.4) -- (.1,1.4) ;
\draw[-]  (-.1,2) -- (.1,2) ;
\draw[-]  (-.1,2.6) -- (.1,2.6) ;
\draw[-]  (-.1,3.2) -- (.1,3.2) ;
\draw[-]  (-.1,3.8) -- (.1,3.8) ;
\draw[-]  (-.1,4.4) -- (.1,4.4) ;
\draw[-]  (-.1,5) -- (.1,5) ;
\draw[-]  (-.1,5.6) -- (.1,5.6) ;
\draw[-]  (-.1,6.2) -- (.1,6.2) ;
\draw[-]  (-.1,6.8) -- (.1,6.8) ;
\draw[-]  (-.1,7.4) -- (.1,7.4) ;
\draw[-]  (-.1,8) -- (.1,8) ;
\draw[-]  (-.1,8.6) -- (.1,8.6) ;
\draw[-]  (-.1,9.2) -- (.1,9.2) ;
\draw[-]  (-.1,9.8) -- (.1,9.8) ;
\draw[-]  (-.1,10.4) -- (.1,10.4) ;
\draw[-]  (-.1,11) -- (.1,11) ;
\draw[-]  (-.1,11.6) -- (.1,11.6) ;
\draw[-]  (-.1,12.2) -- (.1,12.2) ;
\draw[-]  (-.1,12.8) -- (.1,12.8) ;
\draw[-]  (-.1,13.4) -- (.1,13.4) ;
\draw[-]  (-.1,14) -- (.1,14) ;
\draw[-]  (-.1,14.6) -- (.1,14.6) ;
\draw[-]  (-.1,15.2) -- (.1,15.2) ;
\draw[-]  (-.1,15.8) -- (.1,15.8) ;
\draw[-]  (-.1,16.4) -- (.1,16.4) ;
\draw[-]  (-.1,17) -- (.1,17) ;
\draw[-]  (-.1,17.6) -- (.1,17.6) ;
\draw[-]  (-.1,18.2) -- (.1,18.2) ;
\draw[-]  (-.1,18.8) -- (.1,18.8) ;
\draw[-]  (-.1,19.4) -- (.1,19.4) ;

\draw[draw=black,thick,fill=red!20]  (0, 0.2) rectangle node {15(1)} (3, 5) ;
\draw[draw=black,thick,fill=red!20]  (3, 0.2) rectangle node {15(2)} (6, 5) ;
\draw[draw=black,thick,fill=red!20]  (6, 0.2) rectangle node {15(3)} (9, 5) ;
\draw[draw=black,thick,fill=red!20]  (9, 0.2) rectangle node {15(4)} (12, 5) ;

\draw[draw=black,thick,fill=green!20]  (0, 5) rectangle node {14(1)} (4, 7.4) ;
\draw[draw=black,thick,fill=green!20]  (4, 5) rectangle node {14(2)} (8, 7.4) ;
\draw[draw=black,thick,fill=green!20]  (8, 5) rectangle node {14(3)} (12, 7.4) ;
\draw[draw=black,thick,fill=green!20]  (0, 7.4) rectangle node {13(1)} (4, 9.8) ;
\draw[draw=black,thick,fill=green!20]  (4, 7.4) rectangle node {13(2)} (8, 9.8) ;
\draw[draw=black,thick,fill=green!20]  (8, 7.4) rectangle node {13(3)} (12, 9.8) ;

\draw[draw=black,thick,fill=yellow!20]  (0, 9.8) rectangle node {12(1)} (6, 11) ;
\draw[draw=black,thick,fill=yellow!20]  (6, 9.8) rectangle node {12(2)} (12, 11) ;
\draw[draw=black,thick,fill=yellow!20]  (0, 11) rectangle node {11(1)} (6, 12.2) ;
\draw[draw=black,thick,fill=yellow!20]  (6, 11) rectangle node {11(2)} (12, 12.2) ;
\draw[draw=black,thick,fill=yellow!20]  (0, 12.2) rectangle node {10(1)} (6, 13.4) ;
\draw[draw=black,thick,fill=yellow!20]  (6, 12.2) rectangle node {10(2)} (12, 13.4) ;
\draw[draw=black,thick,fill=yellow!20]  (0, 13.4) rectangle node {9(1)} (6, 14.6) ;
\draw[draw=black,thick,fill=yellow!20]  (6, 13.4) rectangle node {9(2)} (12, 14.6) ;

\draw[draw=black,thick,fill=blue!20]  (0, 14.6) rectangle node {8} (12, 15.2) ;
\draw[draw=black,thick,fill=blue!20]  (0, 15.2) rectangle node {7} (12, 15.8) ;
\draw[draw=black,thick,fill=blue!20]  (0, 15.8) rectangle node {6} (12, 16.4) ;
\draw[draw=black,thick,fill=blue!20]  (0, 16.4) rectangle node {5} (12, 17) ;
\draw[draw=black,thick,fill=blue!20]  (0, 17) rectangle node {4} (12, 17.6) ;
\draw[draw=black,thick,fill=blue!20]  (0, 17.6) rectangle node {3} (12, 18.2) ;
\draw[draw=black,thick,fill=blue!20]  (0, 18.2) rectangle node {2} (12, 18.8) ;
\draw[draw=black,thick,fill=blue!20]  (0, 18.8) rectangle node {1} (12, 19.4) ;

\draw[-]  (0,-.1) node[below] {\scriptsize{$0$}} -- (0,.1) ;
\draw[-]  (3,-.1) node[below] {\scriptsize{$\frac{1}{4}$}} -- (3,.1) ;
\draw[-]  (6,-.1) node[below] {\scriptsize{$\frac{1}{2}$}} -- (6,.1) ;
\draw[-]  (9,-.1) node[below] {\scriptsize{$\frac{3}{4}$}} -- (9,.1) ;
\draw[-]  (12,-.1) node[below] {\scriptsize{$1$}} -- (12,.1) ;
\end{tikzpicture}
}
\hspace{0.7cm}
\subfloat[][]{
\begin{tikzpicture}[xscale=0.3,yscale=.3]
\scriptsize
\draw[->] (0,0) -- (16,0) node[pos=.98,above] {time} ;
\draw[->] (0,0) -- (0,20.3) ; 
\draw (-.5,9.7) node[rotate=90] {Processors} ;
\draw[-]  (-.1,0.8) -- (.1,0.8) ;
\draw[-]  (-.1,1.4) -- (.1,1.4) ;
\draw[-]  (-.1,2) -- (.1,2) ;
\draw[-]  (-.1,2.6) -- (.1,2.6) ;
\draw[-]  (-.1,3.2) -- (.1,3.2) ;
\draw[-]  (-.1,3.8) -- (.1,3.8) ;
\draw[-]  (-.1,4.4) -- (.1,4.4) ;
\draw[-]  (-.1,5) -- (.1,5) ;
\draw[-]  (-.1,5.6) -- (.1,5.6) ;
\draw[-]  (-.1,6.2) -- (.1,6.2) ;
\draw[-]  (-.1,6.8) -- (.1,6.8) ;
\draw[-]  (-.1,7.4) -- (.1,7.4) ;
\draw[-]  (-.1,8) -- (.1,8) ;
\draw[-]  (-.1,8.6) -- (.1,8.6) ;
\draw[-]  (-.1,9.2) -- (.1,9.2) ;
\draw[-]  (-.1,9.8) -- (.1,9.8) ;
\draw[-]  (-.1,10.4) -- (.1,10.4) ;
\draw[-]  (-.1,11) -- (.1,11) ;
\draw[-]  (-.1,11.6) -- (.1,11.6) ;
\draw[-]  (-.1,12.2) -- (.1,12.2) ;
\draw[-]  (-.1,12.8) -- (.1,12.8) ;
\draw[-]  (-.1,13.4) -- (.1,13.4) ;
\draw[-]  (-.1,14) -- (.1,14) ;
\draw[-]  (-.1,14.6) -- (.1,14.6) ;
\draw[-]  (-.1,15.2) -- (.1,15.2) ;
\draw[-]  (-.1,15.8) -- (.1,15.8) ;
\draw[-]  (-.1,16.4) -- (.1,16.4) ;
\draw[-]  (-.1,17) -- (.1,17) ;
\draw[-]  (-.1,17.6) -- (.1,17.6) ;
\draw[-]  (-.1,18.2) -- (.1,18.2) ;
\draw[-]  (-.1,18.8) -- (.1,18.8) ;
\draw[-]  (-.1,19.4) -- (.1,19.4) ;

\draw[draw=black,thick,fill=red!20]  (0, 0.2) rectangle node {15(1)} (6, 1.4) ;
\draw[draw=black,thick,fill=green!20]  (0, 1.4) rectangle node {14(1)} (6, 2.6) ;
\draw[draw=black,thick,fill=green!20]  (0, 2.6) rectangle node {13(1)} (6, 3.8) ;
\draw[draw=black,thick,fill=yellow!20]  (0, 3.8) rectangle node {12(1)} (6, 5) ;
\draw[draw=black,thick,fill=yellow!20]  (0, 5) rectangle node {11(1)} (6, 6.2) ;
\draw[draw=black,thick,fill=yellow!20]  (0, 6.2) rectangle node {10(1)} (6, 7.4) ;
\draw[draw=black,thick,fill=yellow!20]  (0, 7.4) rectangle node {9(1)} (6, 8.6) ;
\draw[draw=black,thick,fill=blue!20]  (0, 8.6) rectangle node {8} (6, 9.8) ;
\draw[draw=black,thick,fill=blue!20]  (0, 9.8) rectangle node {7} (6, 11) ;
\draw[draw=black,thick,fill=blue!20]  (0, 11) rectangle node {6} (6, 12.2) ;
\draw[draw=black,thick,fill=blue!20]  (0, 12.2) rectangle node {5} (6, 13.4) ;
\draw[draw=black,thick,fill=blue!20]  (0, 13.4) rectangle node {4} (6, 14.6) ;
\draw[draw=black,thick,fill=blue!20]  (0, 14.6) rectangle node {3} (6, 15.8) ;
\draw[draw=black,thick,fill=blue!20]  (0, 15.8) rectangle node {2} (4.64, 17.6) ;
\draw[draw=black,thick,fill=blue!20]  (0, 17.6) rectangle node {1} (4.64, 19.4) ;

\draw[draw=black,thick,fill=red!20]  (6, 0.2) rectangle node {15(2)} (10, 2.6) ;
\draw[draw=black,thick,fill=green!20]  (6, 2.6) rectangle node {14(2)} (10, 5) ;
\draw[draw=black,thick,fill=green!20]  (6, 5) rectangle node {13(2)} (10, 7.4) ;
\draw[draw=black,thick,fill=yellow!20]  (6, 7.4) rectangle node {12(2)} (9.61, 10.4) ;
\draw[draw=black,thick,fill=yellow!20]  (6, 10.4) rectangle node {11(2)} (9.61, 13.4) ;
\draw[draw=black,thick,fill=yellow!20]  (6, 13.4) rectangle node {10(2)} (9.61, 16.4) ;
\draw[draw=black,thick,fill=yellow!20]  (6, 16.4) rectangle node {9(2)} (9.61, 19.4) ;

\draw[draw=black,thick,fill=red!20]  (10, 0.2) rectangle node {15(3)} (12.78, 6.2) ;
\draw[draw=black,thick,fill=green!20]  (10, 6.2) rectangle node {14(3)} (12.64, 12.8) ;
\draw[draw=black,thick,fill=green!20]  (10, 12.8) rectangle node {13(3)} (12.64, 19.4) ;

\draw[draw=black,thick,fill=red!20]  (12.78, 0.2) rectangle node {15(4)} (14.78, 19.4) ;

\draw[-]  (0,-.1) node[below] {\scriptsize{$t_0\!\!=\!\!0$}} -- (0,.1) ;
\draw[-]  (6,-.1) node[below] {\scriptsize{$t_1\!\!=\!\!\frac{1}{2}$}} -- (6,.1) ;
\draw[-]  (10,-.1) node[below] {\scriptsize{$t_2\!\!=\!\!\frac{5}{6}$}} -- (10,.1) ;
\draw[-]  (12.78,-.05) node[below] {\scriptsize{$t_3\!\!\approx \!\!1.07$}} -- (12.78,.1) ;
\draw[-]  (14.78,-.6) node[below] {\scriptsize{$t_4\!\! \approx \!\!1.23$}} -- (14.78,.1) ;
\end{tikzpicture}
}
\caption{For the lower bound instance of Figure \ref{fig:work3}: (a) An offline schedule with a makespan of 1; (b) An online algorithm's schedule, allocating (approximately)
the same number of processors to all linear chains and producing a makespan of \mbox{$t_4\approx 1.23$}.}
\label{fig:work1}
\end{figure*}

First, the optimal offline algorithm could schedule the tasks as follows: for any group $i \in [1,K]$,
it allocates $2^{i-1}$ processors to each linear chain in the group.
The total number of required processors is then $\sum_{i=1}^K 2^{i-1} \times 2^{K-i}=K\times 2^{K-1}=P$.
Thus, all linear chains could be executed in parallel. Furthermore, they will all be completed at time 1,
since each linear chain in group $i$ has $i$ tasks, and each task has an execution time
$t(2^{i-1})=\frac{1}{\lg(2^{i-1})+1}=\frac{1}{i}$.
Figure~\ref{fig:work1}(a) illustrates the schedule for this instance with $\ell = 2$.

Now, we establish a lower bound on the makespan of any deterministic online algorithm.
For any $i \in [1,K-1]$, let $L_i$ denote the set of linear chains in all groups $j\le i$,
and let $L'_i$ denote the set of linear chains in all groups $j > i$. Let us define $t_i$ to be
the first time a linear chain in $L'_i$ completes $i$ tasks. We further define $t_0=0$ and
let $t_K$ denote the makespan of the online algorithm.

\begin{lemma}\label{lem.ti}
Any deterministic online algorithm could produce a schedule that satisfies:
$t_i-t_{i-1}\geq \frac{1}{\ell+i},~\forall i \in [1,K]$.
\end{lemma}

\begin{proof}
Since all tasks are identical, an online algorithm cannot distinguish the linear chains.
Thus, for any $i \in [1, K]$, an adversary could make all linear chains that first complete $i$ tasks
by the online algorithm be chains from $L_i$.
Therefore, at time $t_{i}$, all linear chains containing exactly $i$ tasks (i.e., the ones from group $i$)
are already completed, and at time~$t_{i-1}$, no linear chain has started its $i$-th task by definition
(this also holds for $t_0$ and $t_K$).
Hence, all tasks in the $i$-th position of the linear chains in group $i$ must be entirely processed between $t_{i}$ and $t_{i-1}$, and the number of such tasks is~$2^{K-i}$.

For the sake of contradiction, suppose we have $t_{i} - t_{i-1} < \frac{1}{\ell+i}$. Thus, the execution time of these tasks must satisfy $t(p)=\frac{1}{\lg(p)+1} \le \frac{1}{\ell+i}$, hence their processor allocation must be at least $p \ge 2^{\ell+i-1} = K\cdot 2^{i-1}$. As the area of the task $a(p) = p t(p) = \frac{p}{\lg(p)+1}$ is increasing with the number of processors, the total area of all tasks that needs to be processed between $t_i$ and $t_{i-1}$ is at least $2^{K-i} \cdot a(K\times 2^{i-1})=\frac{2^{K-i} \cdot K\cdot 2^{i-1}}{\lg(K\cdot 2^{i-1})+1}=\frac{K\cdot 2^{K-1}}{\ell+i}=\frac{P}{\ell+i}$. Since we have $P$ processors, the total time required to process this area is at least $\frac{1}{\ell+i}$, which contradicts $t_{i} - t_{i-1} < \frac{1}{\ell+i}$.
\end{proof}

One strategy to cope with the worst-case scenario above is to allocate the same number of processors to each linear chain (or more precisely allocate one more processor to some linear chains in order to utilize all the processors). Figure~\ref{fig:work1}(b) illustrates this strategy for the same instance with $\ell=2$.

Finally, we can use the result of Lemma \ref{lem.ti} to lower bound the makespan of an online algorithm, which is given by $t_K=\sum_{i=1}^K (t_i-t_{i-1})$. Since for all $j$, $\ln(j)+\gamma<\sum_{i=1}^j \frac{1}{i}<\ln(j)+\gamma+\frac{1}{j}$ where $\gamma$ is the Euler constant, we obtain:

{\small{\begin{align*}
t_K &\geq \sum_{i=1}^K \frac{1}{\ell+i} > \sum_{i=\ell+1}^K \frac{1}{i} = \sum_{i=1}^K \frac{1}{i} - \sum_{i=1}^{\ell} \frac{1}{i} \\
 &> \left(\ln(K)+\gamma\right)-\left(\ln(\ell)+\gamma+\frac{1}{\ell}\right) =\ln(K)-\ln(\ell)-\frac{1}{\ell} \ .  \qedhere
\end{align*}}}
\end{proof}

\section{Conclusion and Future Work}
\label{sec.conclusion}

In this paper, we have studied the online scheduling of moldable task graphs to minimize makespan with tasks obeying several common speedup models.
To the best of our knowledge, no competitive ratio was known under this setting, except for
the roofline model~\cite{Feldmann98_DAG}. Owing to the design of a new online algorithm and a novel analysis framework,
we have extended the result and derived competitive ratios for several
other speedup models, including the communication model, the Amdahl's model and a general combination.
We have also shown that no online list scheduling algorithm with deterministic local decisions for processor allocation may have a better competitive ratio than ours for the roofline, communication and Amdahl's models.
Finally, we have considered the arbitrary speedup model and established a lower bound for any deterministic online algorithm. Altogether, these new results
lay the foundations for further study of this important but difficult scheduling problem.

For future work, we will consider extending the algorithm and analysis to other common speedup models.
We also plan to extend our algorithm and analysis to other online scheduling settings (e.g., for independent tasks released over time, and for special task graphs such as fork-join graphs or trees).
Finally, we will expand this study to a more practical side by experimentally benchmarking the performance of our algorithm using realistic workflows.

\section*{Acknowledgements}
We thank Anne Benoit and Yves Robert for many helpful discussions and for their inputs to the preliminary version of this paper \cite{Benoit22_online}.
This work is supported in part by the US National Science Foundation grant \#2135310.

\bibliographystyle{abbrv}

\begin{thebibliography}{10}

\bibitem{Agrawal10_dynamic}
K.~Agrawal, C.~E. Leiserson, and J.~Sukha.
\newblock Executing task graphs using work-stealing.
\newblock In {\em IPDPS}, pages 1--12, 2010.

\bibitem{Amdahl67}
G.~M. Amdahl.
\newblock Validity of the single processor approach to achieving large scale
  computing capabilities.
\newblock In {\em AFIPS'67}, pages 483--485, 1967.

\bibitem{Belkhale90}
K.~P. Belkhale and P.~Banerjee.
\newblock An approximate algorithm for the partitionable independent task
  scheduling problem.
\newblock In {\em ICPP}, pages 72--75, 1990.

\bibitem{Belkhale91_DAG}
K.~P. Belkhale, P.~Banerjee, and W.~S. Av.
\newblock A scheduling algorithm for parallelizable dependent tasks.
\newblock In {\em IPPS}, pages 500--506, 1991.

\bibitem{Benoit20_cluster}
A.~Benoit, V.~{Le F\`evre}, L.~Perotin, P.~Raghavan, Y.~Robert, and H.~Sun.
\newblock Resilient scheduling of moldable jobs on failure-prone platforms.
\newblock In {\em {IEEE Cluster}}, 2020.

\bibitem{Benoit21_ieeetc}
A.~Benoit, V.~{Le F\`evre}, L.~Perotin, P.~Raghavan, Y.~Robert, and H.~Sun.
\newblock Resilient scheduling of moldable parallel jobs to cope with silent
  errors.
\newblock {\em IEEE Transactions on Computers}, 71(07):1696--1710, 2022.

\bibitem{Benoit22_online}
A.~Benoit, L.~Perotin, Y.~Robert, and H.~Sun.
\newblock Online scheduling of moldable task graphs under common speedup
  models.
\newblock In {\em ICPP}, pages 51:1--51:11, 2022.

\bibitem{Blazewicz01}
J.~Blazewicz, M.~Machowiak, G.~Mouni{\'e}, and D.~Trystram.
\newblock Approximation algorithms for scheduling independent malleable tasks.
\newblock In {\em Euro-Par}, pages 191--197, 2001.

\bibitem{Canon20_online}
L.~Canon, L.~Marchal, B.~Simon, and F.~Vivien.
\newblock Online scheduling of task graphs on heterogeneous platforms.
\newblock {\em {IEEE} Trans. Parallel Distributed Syst.}, 31(3):721--732, 2020.

\bibitem{Chen13_concave}
C.-Y. Chen and C.-P. Chu.
\newblock A 3.42-approximation algorithm for scheduling malleable tasks under
  precedence constraints.
\newblock {\em IEEE Trans. Parallel Distrib. Syst.}, 24(8):1479–1488, 2013.

\bibitem{Dutton07_ECT}
R.~A. Dutton and W.~Mao.
\newblock Online scheduling of malleable parallel jobs.
\newblock In {\em PDCS}, pages 136--141, 2007.

\bibitem{Feitelson96}
D.~G. Feitelson and L.~Rudolph.
\newblock Toward convergence in job schedulers for parallel supercomputers.
\newblock In {\em Job Scheduling Strategies for Parallel Processing}, pages
  1--26. Springer, 1996.

\bibitem{Feldmann98_DAG}
A.~Feldmann, M.-Y. Kao, J.~Sgall, and S.-H. Teng.
\newblock Optimal on-line scheduling of parallel jobs with dependencies.
\newblock {\em Journal of Combinatorial Optimization}, 1(4):393--411, 1998.

\bibitem{Graham69}
R.~L. Graham.
\newblock Bounds on multiprocessing timing anomalies.
\newblock {\em SIAM Journal on Applied Mathematics}, 17(2):416--429, 1969.

\bibitem{Havill08_SET}
J.~T. Havill and W.~Mao.
\newblock Competitive online scheduling of perfectly malleable jobs with setup
  times.
\newblock {\em European Journal of Operational Research}, 187:1126--1142, 2008.

\bibitem{Jansen12_3over2}
K.~Jansen.
\newblock A ($3/2+\epsilon$) approximation algorithm for scheduling moldable
  and non-moldable parallel tasks.
\newblock In {\em SPAA}, pages 224--235, 2012.

\bibitem{Jansen18_PTAS}
K.~{Jansen} and F.~{Land}.
\newblock Scheduling monotone moldable jobs in linear time.
\newblock In {\em IPDPS}, pages 172--181, 2018.

\bibitem{Jansen10_PTAS}
K.~Jansen and R.~Thöle.
\newblock Approximation algorithms for scheduling parallel jobs.
\newblock {\em SIAM Journal on Computing}, 39(8):3571--3615, 2010.

\bibitem{Jansen05_concave}
K.~Jansen and H.~Zhang.
\newblock Scheduling malleable tasks with precedence constraints.
\newblock In {\em SPAA}, page 86–95, 2005.

\bibitem{Jansen06_DAG}
K.~Jansen and H.~Zhang.
\newblock An approximation algorithm for scheduling malleable tasks under
  general precedence constraints.
\newblock {\em ACM Trans. Algorithms}, 2(3):416--434, 2006.

\bibitem{Johannes06_list}
B.~Johannes.
\newblock Scheduling parallel jobs to minimize the makespan.
\newblock {\em J. of Scheduling}, 9(5):433--452, 2006.

\bibitem{JOHNSON96_dynamic}
T.~Johnson, T.~A. Davis, and S.~M. Hadfield.
\newblock A concurrent dynamic task graph.
\newblock {\em Parallel Computing}, 22(2):327--333, 1996.

\bibitem{Kell15_Improved}
N.~Kell and J.~Havill.
\newblock Improved upper bounds for online malleable job scheduling.
\newblock {\em J. of Scheduling}, 18(4):393--410, 2015.

\bibitem{Lepere01_DAG}
R.~Lep\`{e}re, D.~Trystram, and G.~J. Woeginger.
\newblock Approximation algorithms for scheduling malleable tasks under
  precedence constraints.
\newblock In {\em ESA}, pages 146--157, 2001.

\bibitem{Ludwig94}
W.~Ludwig and P.~Tiwari.
\newblock Scheduling malleable and nonmalleable parallel tasks.
\newblock In {\em SODA}, pages 167--176, 1994.

\bibitem{Mounie99_sqrt3}
G.~Mouni\'{e}, C.~Rapine, and D.~Trystram.
\newblock Efficient approximation algorithms for scheduling malleable tasks.
\newblock In {\em SPAA}, pages 23--32, 1999.

\bibitem{Mounie07_dualapprox}
G.~Mouni\'{e}, C.~Rapine, and D.~Trystram.
\newblock A 3/2-approximation algorithm for scheduling independent monotonic
  malleable tasks.
\newblock {\em SIAM J. Comput.}, 37(2):401--412, 2007.

\bibitem{Sleator1985}
D.~D. Sleator and R.~E. Tarjan.
\newblock Amortized efficiency of list update and paging rules.
\newblock {\em Commun. ACM}, 28(2):202–208, 1985.

\bibitem{Turek92}
J.~Turek, J.~L. Wolf, and P.~S. Yu.
\newblock Approximate algorithms scheduling parallelizable tasks.
\newblock In {\em SPAA}, 1992.

\bibitem{Ullman75_NPcomplete}
J.~D. Ullman.
\newblock Np-complete scheduling problems.
\newblock {\em J. Comput. Syst. Sci.}, 10(3):384--393, 1975.

\bibitem{Wang92_DAG}
Q.~Wang and K.~H. Cheng.
\newblock A heuristic of scheduling parallel tasks and its analysis.
\newblock {\em SIAM J. Comput.}, 21(2):281--294, 1992.

\bibitem{Williams2009}
S.~Williams, A.~Waterman, and D.~Patterson.
\newblock Roofline: An insightful visual performance model for multicore
  architectures.
\newblock {\em Commun. ACM}, 52(4):65–76, 2009.

\bibitem{Ye18_online}
D.~Ye, D.~Z. Chen, and G.~Zhang.
\newblock Online scheduling of moldable parallel tasks.
\newblock {\em J. of Scheduling}, 21(6):647--654, 2018.

\end{thebibliography}

\end{document}